\documentclass[12pt, draftcls, onecolumn]{IEEEtran}
\ifCLASSINFOpdf
  \usepackage[pdftex]{graphicx}
\else
\fi
\usepackage[cmex10]{amsmath}
\usepackage[tight,footnotesize]{subfigure}
\usepackage{amssymb}
\usepackage{mathtools}
\usepackage{tikz}
\usetikzlibrary{shapes.symbols,patterns}
\usetikzlibrary{shapes,shapes.multipart,shapes.geometric}
\usetikzlibrary{arrows.meta}

\usepackage{multirow}
\usepackage{epsfig}
\usepackage{epstopdf}
\usepackage{subfigure}
\usepackage{amsthm}
\usepackage{amssymb}
\usepackage{amsmath}
\usepackage{multirow}

\usepackage[linesnumbered,ruled,norelsize]{algorithm2e}
	\makeatletter
	\newcommand{\removelatexerror}{\let\@latex@error\@gobble}
	\makeatother

\usepackage{booktabs}
\usepackage{pgfplots}
\usepackage{epstopdf}
\usepackage{scalefnt}
\usepackage{adjustbox}
\usepackage{ifthen}
\usetikzlibrary{shapes,shapes.geometric}
\usetikzlibrary{calc,backgrounds}
\tikzset{pic shift/.store in=\shiftcoord,
	pic shift = {(0,0)},
	pics/mySensor/.style args={scale #1 RFrange #2 Number #3 Opacity #4}{
		code={		
		\tikzset{cyl/.style={aspect=2,minimum height=0.1cm,minimum width=1.5cm,cylinder,draw=black,line width=2,shape border rotate=90,cylinder uses custom fill, cylinder body fill=blue!10}}			
			
			\begin{scope}[shift={\shiftcoord},scale=#1, every node/.append style={scale=#1}]
			\node[cyl] (-cyl) {};
			\ifthenelse{#3=0}{}{
			\node at ([yshift=-1]-cyl) {\scalefont{1.2} \textbf{#3}};}
			\draw[blue!70!black,very thick] (0,0.4) -- (0.75,0.4);
			\draw[fill=blue!70!black,draw=blue] (-0.1,0.5) -- (0.3,0.5) -- (0.1,0.3) -- (-0.3,0.3) -- cycle ;
			\foreach \i in {1,...,3}{
				\draw[blue!70!black,very thick] (-0.05+\i*0.1,0.55) -- (-0.35+\i *0.1,0.25);
			}
			 
			\draw[blue!70!black,line width=4,line cap=round] (0.75,0.4) -- (0.75,2.5-1);
			
			\draw[fill=blue!70!black,blue!70!black,line width=4] (0.75,1.5) -- (0.75+0.29,1.5+0.5) -- (0.75-0.29,1.5+0.5) -- cycle;
			
			\ifthenelse{#2=0}{}{
			\begin{scope}[on background layer]
			\draw[very thick,red!40!green,fill=red!40!green,opacity=#4] (0,0) ellipse ({#2*2} and #2);
			\draw[line width=1mm,dashed,black] (0,0) ellipse ({#2*2} and #2);
			\end{scope}
			}	
			\end{scope}
			
			}
	}
}
\tikzset{pic shift/.store in=\shiftcoord,
	pic shift = {(0,0)},
	pics/myManifold/.style args={scale #1}{
		code={	
			\begin{scope}[shift={\shiftcoord},scale=#1, every node/.append style={scale=#1}]

			\pgfmathsetmacro{\px}{-0.0682}
			\pgfmathsetmacro{\py}{-0.2079}
			\pgfmathsetmacro{\pz}{0.9758}
        	\begin{axis}[hide axis]
			\addplot3[samples=21,
            	domain=-20:20,
            	y domain=70:110, 
            	colormap/greenyellow,           	
				scatter/use mapped color={
					draw=mapped color,
					fill=mapped color
				},				
				surf,
				shader=faceted interp,opacity = 1, 
				z buffer=sort,draw=green!20]({sin(x)*sin(y)},{cos(y)}, {cos(x)*sin(y)});	
			\addplot3[samples=2,
            	domain=-0.2:0.1,
            	y domain=-0.4:0.2,
            	colormap/blackwhite,           	
				scatter/use mapped color={
					draw=black,
					fill=none
				},
				surf,
				opacity = 0, 
				z buffer=sort,fill=white,draw=black,very thick]({x},{y}, {1-(\px/\pz)*x-(\py/\pz)*y});
				
			\end{axis}
			
			\end{scope}
			
			}
	}
}
\tikzset{pic shift/.store in=\shiftcoord,
	pic shift = {(0,0)},
	pics/myManifoldv2/.style args={scale #1 retraction #2 number #3}{
		code={	
			\begin{scope}[shift={\shiftcoord},scale=#1, every node/.append style={scale=#1}]

			\begin{scope}[on background layer]
				\path (0,0) pic {myManifold={scale 1.6}};
			\end{scope}
			\draw[black!80!green, very thick] (1.6,1.8) -- (4.8,1.8) -- (8.2,7.4) -- (5,7.4) -- cycle;
			
			\ifthenelse{#2=0}{}
			{
			\draw[black!50!red,thick,dashed]  (5.8,6.4) -- (5.8,5.2);
			\node[anchor=north,xshift=0.2cm] at (5.8,5.2) {$R_{\mathbf{Y}}(\mathbf{B})$};
			\fill[black!50!red] (5.8,5.2) circle [radius=2pt];
			\node[rectangle,draw=black!50!red,thick,inner sep=0.1cm] at (5.7,5.3) {};
			}
			\draw[black!50!red,very thick] (3.4,3.4) .. controls (3.65,4.8) and (5.4,5.2) .. (5.8,5.2);
			\node[black!50!red] at (3.8,3.4) {$\gamma(t)$};
			\draw[black!80!green,very thick,->] (3.8,4.25) -- (5.8,6.4);	
			\ifthenelse{#2=0}
			{
			\node[anchor=south west,xshift=-0.6cm,yshift=-0.2cm] at (5.8,6.4) {$\mathbf{B}=\gamma^\prime (t)\bigg|_{t=0}$};
			\draw[black!80!green,very thick,->] (0.2,3.6) -- (0.2,7);
			\node[anchor=east,xshift=-0.1cm] at (0.2,7) {$t\in\mathbb{R}$};

			\draw[black!80!green,very thick] (0.3,5.4) -- (0.1,5.4);
			\node[anchor=east] at (0.1,5.4) {0};
	
			\draw[black!80!green,thick,dotted] (0.2,5.4) -- (3.8,4.25);
			\draw[black!80!green,thick,dotted] (0.2,3.8) -- (3.5,3.8);
			\draw[black!80!green,thick,dotted] (0.2,6.4) -- (4.2,4.6);
			}
			{
			\node[anchor=south west,xshift=-0.6cm,yshift=-0.2cm] at (5.8,6.4) {$\mathbf{B}$};
			}
			\node[black] at (4.1,4.1) {$\mathbf{Y}$};
			\fill[black!50!red] (3.8,4.25) circle [radius=2pt];

			\node at (8.2,1) {\textcolor{black!70!green}{$\widetilde{\mathcal{Y}}$}};
			\node at (2.4,2.2) {$T_{\mathbf{Y}}\widetilde{\mathcal{Y}}$};
			
			\node at (4.8,0.4) {(#3)};
			\end{scope}
			
			}
	}
}

\newtheorem{theorem}{Theorem}

\newtheorem{lemma}[theorem]{Lemma}

\newcommand{\bi}{\begin{itemize}}
\newcommand{\ei}{\end{itemize}}
\newcommand{\be}{\begin{eqnarray}}
\newcommand{\ee}{\end{eqnarray}}
\newcommand{\efr}{\end{frame}}
\newcommand{\bc}{\begin{center}}
\newcommand{\ec}{\end{center}}
\newcommand{\diag}{\operatornamewithlimits{diag}}

\usepackage{threeparttable}
\newcommand\Tstrut{\rule{0pt}{2.6ex}}       
\newcommand\Bstrut{\rule[-0.9ex]{0pt}{0pt}} 

\begin{document}

\newcommand{\matc}[2][ccccccccccccccccccc]{\left[
\begin{array}{#1}
#2\\
\end{array}
\right]}
\newcommand{\matr}[2][rrrrrrrrrrrrrrrrrrrrrrrr]{\left[
\begin{array}{#1}
#2\\
\end{array}
\right]}
\newcommand{\matl}[2][lllllllllllllllllll]{\left[
\begin{array}{#1}
#2\\
\end{array}
\right]}

\newcommand{\pt}[2]{P_{T_{\mathbf{Y}_{#1}}\widetilde{\mathcal{Y}}}(#2)}

\newcommand*\circled[1]{\tikz[baseline=(char.base)]{
  \node[shape=circle,draw,inner sep=2pt] (char) {#1};}}


{\Large \textbf{Notice:} This work has been submitted to the IEEE for possible publication. Copyright may be transferred without notice, after which this version may no longer be accessible.}

\clearpage


%
\title{Low-Rank Matrix Completion: A Contemporary Survey}

\author{\IEEEauthorblockN{Luong Trung Nguyen, Junhan Kim, Byonghyo Shim}

\IEEEauthorblockA{Information System Laboratory\\
Department of Electrical and Computer Engineering, Seoul National University\\
Email: \{ltnguyen,junhankim,bshim\}@islab.snu.ac.kr}
}


%


\maketitle

\begin{abstract}
As a paradigm to recover unknown entries of a matrix from partial observations, low-rank matrix completion (LRMC) has generated a great deal of interest. Over the years, there have been lots of works on this topic but it might not be easy to grasp the essential knowledge from these studies. This is mainly because many of these works are highly theoretical or a proposal of new LRMC technique. In this paper, we give a contemporary survey on LRMC. In order to provide better view, insight, and understanding of potentials and limitations of LRMC, we present early scattered results in a structured and accessible way. Specifically, we classify the state-of-the-art LRMC techniques into two main categories and then explain each category in detail. We next discuss issues to be considered when one considers using LRMC techniques. These include intrinsic properties required for the matrix recovery and how to exploit a special structure in LRMC design. We also discuss the convolutional neural network (CNN) based LRMC algorithms exploiting the graph structure of a low-rank matrix. Further, we present the recovery performance and the computational complexity of the state-of-the-art LRMC techniques. Our hope is that this survey article will serve as a useful guide for practitioners and non-experts to catch the gist of LRMC.
\end{abstract}


%
\IEEEpeerreviewmaketitle

\section{Introduction}
\label{sec:intro}
\PARstart{I}{n} the era of big data, the low-rank matrix has become a useful and popular tool to express two-dimensional information. One well-known example is the rating matrix in the recommendation systems representing users' tastes on products~\cite{netflix}. Since users expressing similar ratings on multiple products tend to have the same interest for the new product, columns associated with users sharing the same interest are highly likely to be the same, resulting in the low rank structure of the rating matrix (see Fig.~\ref{fig:fig007}). Another example is the Euclidean distance matrix formed by the pairwise distances of a large number of sensor nodes. Since the rank of a Euclidean distance matrix in the $k$-dimensional Euclidean space is at most $k+2$ (if $k = 2$, then the rank is 4), this matrix can be readily modeled as a low-rank matrix \cite{pal2010,luongITA,luong2019}. 


\begin{figure*} [t!]
\label{fig:fig007}
\centering    
\includegraphics[width=0.9\textwidth]{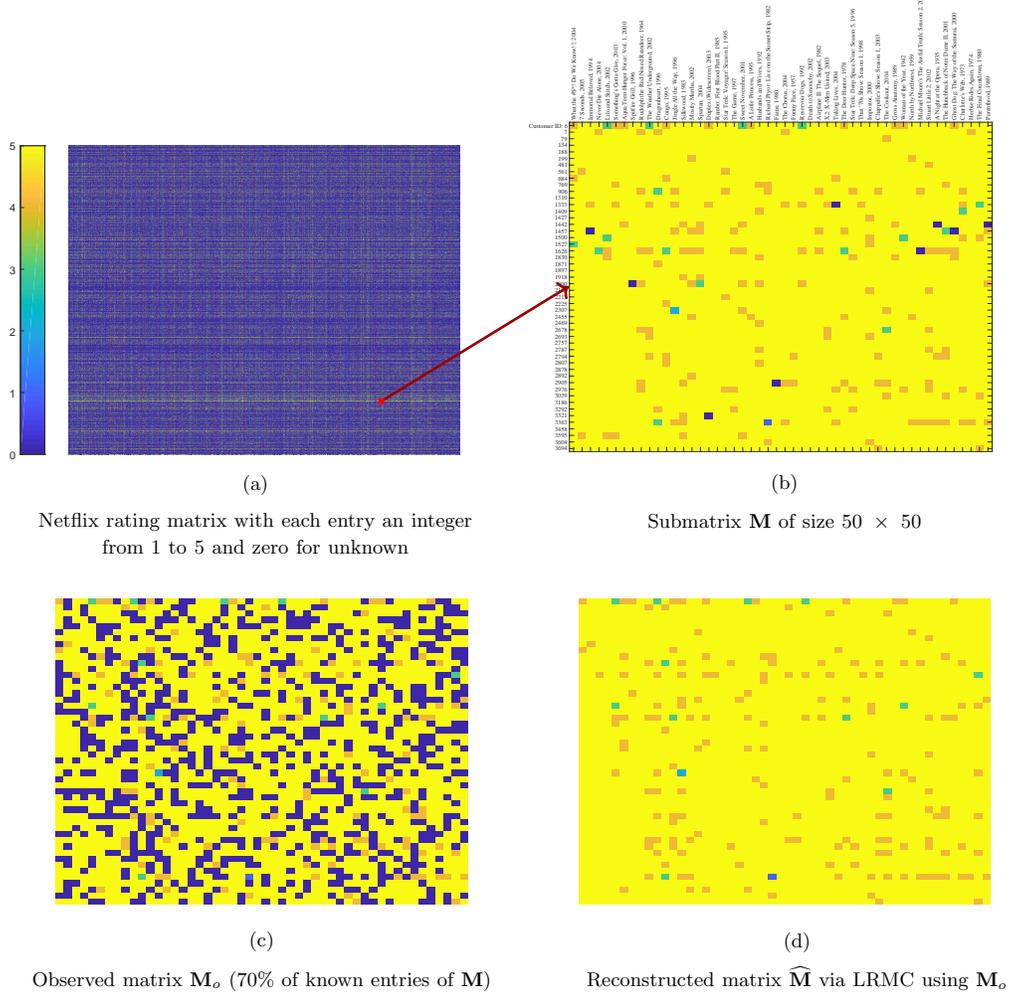}  
\caption {Recommendation system application of LRMC. Entries of $\widehat{\mathbf{M}}$ are then simply rounded to integers, achieving 97.2\% accuracy.} 
\end{figure*} 

One major benefit of the low-rank matrix is that the essential information, expressed in terms of degree of freedom, in a matrix is much smaller than the total number of entries. Therefore, even though the number of observed entries is small, we still have a good chance to recover the whole matrix. There are a variety of scenarios where the number of observed entries of a matrix is tiny. In the recommendation systems, for example, users are recommended to submit the feedback in a form of rating number, e.g., 1 to 5 for the purchased product. However, users often do not want to leave a feedback and thus the rating matrix will have many missing entries. Also, in the internet of things (IoT) network, sensor nodes have a limitation on the radio communication range or under the power outage so that only small portion of entries in the Euclidean distance matrix is available. 

When there is no restriction on the rank of a matrix, the problem to recover unknown entries of a matrix from partial observed entries is ill-posed. This is because any value can be assigned to unknown entries, which in turn means that there are infinite number of matrices that agree with the observed entries. As a simple example, consider the following $2\times 2$ matrix with one unknown entry marked ?
\be
\mathbf{M} = \matc{ 1 & 5 \\ 2 & ?}. 
\label{eq:eq001}
\ee
%
If $\mathbf{M}$ is a full rank, i.e., the rank of $\mathbf{M}$ is two, then any value except $10$ can be assigned to $?$. Whereas, if $\mathbf{M}$ is a low-rank matrix (the rank is one in this trivial example), two columns differ by only a constant and hence unknown element ? can be easily determined using a linear relationship between two columns (? = 10). This example is obviously simple, but the fundamental principle to recover a large dimensional matrix is not much different from this and the \textit{low-rank} constraint plays a central role in recovering unknown entries of the matrix.  

Before we proceed, we discuss a few notable applications where the underlying matrix is modeled as a low-rank matrix.

\begin{enumerate}
\item {\bf Recommendation system}: In 2006, the online DVD rental company Netflix announced a contest to improve the quality of the company's movie recommendation system. The company released a training set of half million customers. Training set contains ratings on more than ten thousands movies, each movie being rated on a scale from $1$ to $5$ \cite{netflix}. The training data can be represented in a large dimensional matrix in which each column represents the rating of a customer for the movies. The primary goal of the recommendation system is to estimate the users' interests on products using the sparsely sampled\footnote{Netflix dataset consists of ratings of more than 17,000 movies by more than 2.5 million users. The number of known entries is only about 1\%~\cite{netflix}.} rating matrix.\footnote{Customers might not necessarily rate all of the movies.} Often users sharing the same interests in key factors such as the type, the price, and the appearance of the product tend to provide the same rating on the movies. The ratings of those users might form a low-rank column space, resulting in the low-rank model of the rating matrix (see Fig. \ref{fig:fig007}). 


\begin{figure*} [t!]
\label{fig:fig1}
\centering    
\includegraphics[width=1.02\textwidth]{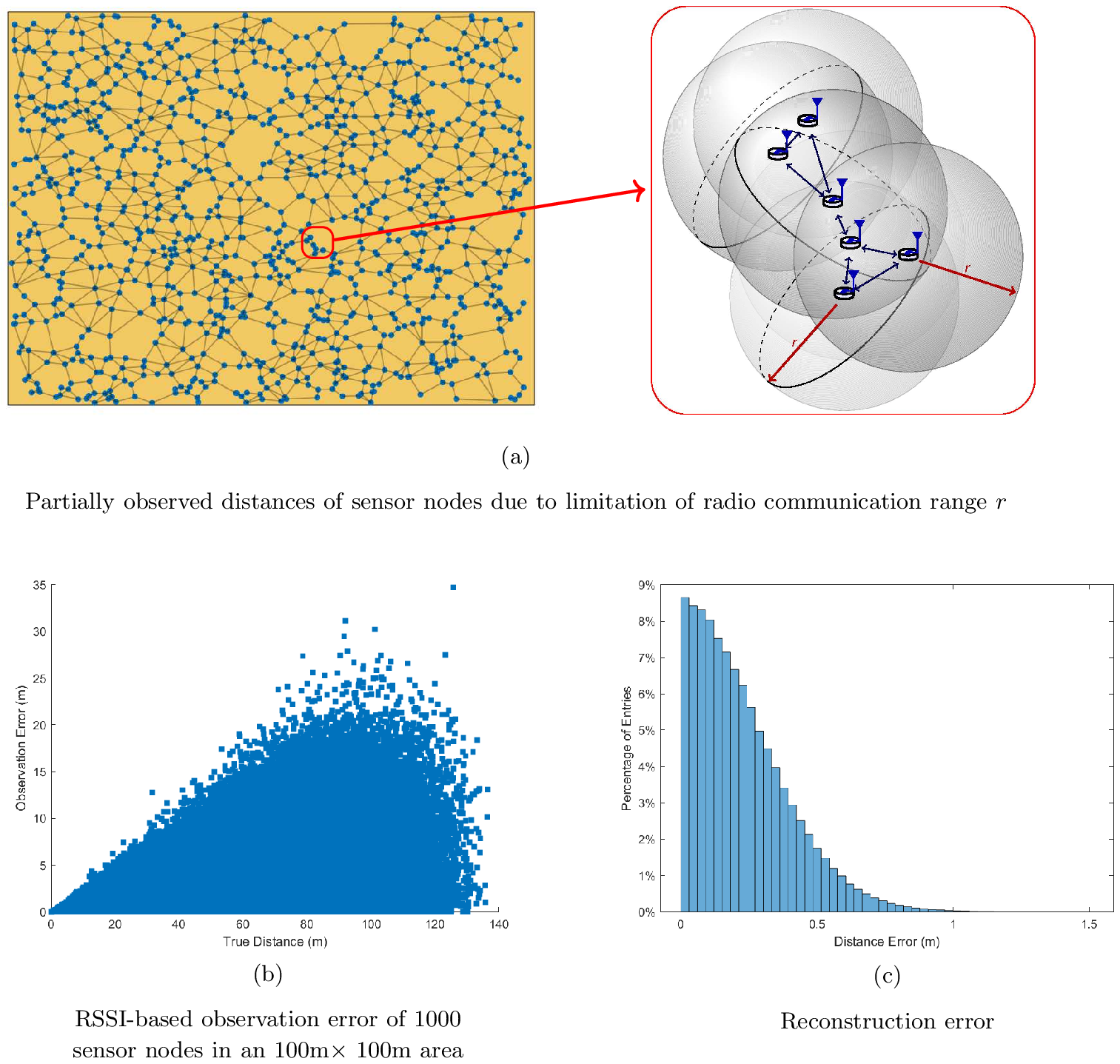} 
\caption {Localization via LRMC \cite{luong2019}. The Euclidean distance matrix can be recovered with 92\% of distance error below 0.5m using 30\% of observed distances.} 
\end{figure*}



\begin{figure*} [t!]
\label{fig:fig006}
\centering    
\includegraphics[width=0.95\textwidth]{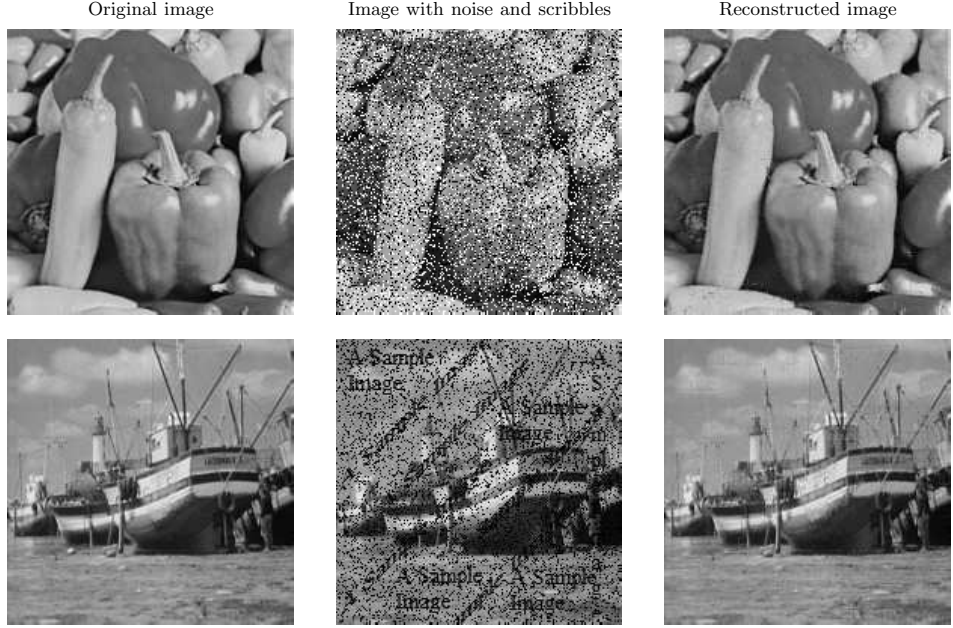}  
\caption {Image reconstruction via LRMC. Recovered images achieve peak SNR $\geq$ 32dB.} 
\end{figure*}


\item {\bf Phase retrieval}: 
The problem to recover a signal not necessarily sparse from the magnitude of its observation is referred to as the phase retrieval. Phase retrieval is an important problem in X-ray crystallography and quantum mechanics since only the magnitude of the Fourier transform is measured in these applications \cite{candes2015}. Suppose the unknown time-domain signal $\mathbf{m} = [ m_{0} \ \cdots \ m_{n-1} ]$ is acquired in a form of the measured magnitude of the Fourier transform. That is,   
\begin{align}
|z_{\omega}| 
&= \frac{1}{\sqrt{n}} \left | \underset{t=0}{\overset{n-1}{\sum}} m_te^{-j2\pi\omega t/n}\right|,~ \omega\in\Omega,
\label{eq:eq100}
\end{align}
where $\Omega$ is the set of sampled frequencies. Further, let 
\begin{align}
\mathbf{f}_{\omega} 
&= \frac{1}{\sqrt{n}} [1 \ e^{-j2\pi\omega/n} \ \cdots \ e^{-j2\pi\omega(n-1)/n}]^{H},
\end{align} 
$\mathbf{M} = \mathbf{mm}^H$ where $\mathbf{m}^{H}$ is the conjugate transpose of $\mathbf{m}$. Then, (\ref{eq:eq100}) can be rewritten as
\begin{align}
|z_{\omega}|^{2} 
&= | \langle \mathbf{f}_{\omega}, \mathbf{m} \rangle |^{2} \\
&= \text{tr}(\mathbf{f}_{\omega}^{H} \mathbf{m} \mathbf{m}^{H} \mathbf{f}_{\omega}) \\
&= \text{tr}(\mathbf{m} \mathbf{m}^{H} \mathbf{f}_{\omega} \mathbf{f}_{\omega}^{H}) \\
&= \langle \mathbf{M},\mathbf{F}_{\omega} \rangle,
\label{eq:eq200}
\end{align}
where $\mathbf{F}_{w} = \mathbf{f}_{w} \mathbf{f}_{w}^{H}$ is the rank-1 matrix of the waveform $\mathbf{f}_\omega$. Using this simple transform, we can express the quadratic magnitude $|z_\omega|^2$ as linear measurement of $\mathbf{M}$. In essence, the phase retrieval problem can be converted to the problem to reconstruct the rank-1 matrix $\mathbf{M}$ in the positive semi-definite (PSD) cone\footnote{If $\mathbf{M}$ is recovered, then the time-domain vector $\mathbf{m}$ can be computed by the eigenvalue decomposition of $\mathbf{M}$.} \cite{candes2015}:
\begin{equation} 
\begin{split}
\min\limits_\mathbf{X} &~~~~~~\text{rank}(\mathbf{X}) \\
\text{subject to}  &~~~~~~\langle \mathbf{M},\mathbf{F}_{\omega} \rangle = |z_\omega|^2,\: \omega\in\Omega\\
&~~~~~~\mathbf{X}\succeq 0.
\end{split}
\end{equation}


\item {\bf Localization in IoT networks}: In recent years, internet of things (IoT) has received much attention for its plethora of applications such as healthcare, automatic metering, environmental monitoring (temperature, pressure, moisture), and surveillance \cite{delamo2015,hackmann2014,pal2010}. Since the action in IoT networks, such as fire alarm, energy transfer, emergency request, is made primarily on the data center, data center should figure out the location information of whole devices in the networks. In this scheme, called network localization (a.k.a. cooperative localization), each sensor node measures the distance information of adjacent nodes and then sends it to the data center. Then the data center constructs a map of sensor nodes using the collected distance information \cite{torgerson1952}. Due to various reasons, such as the power outage of a sensor node or the limitation of radio communication range (see Fig. \ref{fig:fig1}), only small number of distance information is available at the data center. Also, in the vehicular networks, it is not easy to measure the distance of all adjacent vehicles when a vehicle is located at the dead zone. An example of the observed Euclidean distance matrix is  
    \be
    \mathbf{M}_o = \matc{0 & d_{12}^2 & d_{13}^2 & ?  & ?\\
      d_{21}^2 & 0 & ? & ? & ?\\
      d_{31}^2 & ? & 0 & d_{34}^2 & d_{35}^2\\
      ? & ? & d_{43}^2 & 0 & d_{45}^2\\
      ? & ? & d_{53}^2 & d_{54}^2 & 0}\nonumber ,
    \ee
where $d_{ij}$ is the pairwise distance between two sensor nodes $i$ and $j$. Since the rank of Euclidean distance matrix $\mathbf{M}$ is at most $k$+2 in the $k$-dimensional Euclidean space ($k = 2$ or $k = 3$) \cite{luongITA,luong2019}, the problem to reconstruct $\mathbf{M}$ can be well-modeled as the LRMC problem.


\begin{figure*} [t!]
\label{fig:fig001}
\centering    
\includegraphics[width=0.95\textwidth]{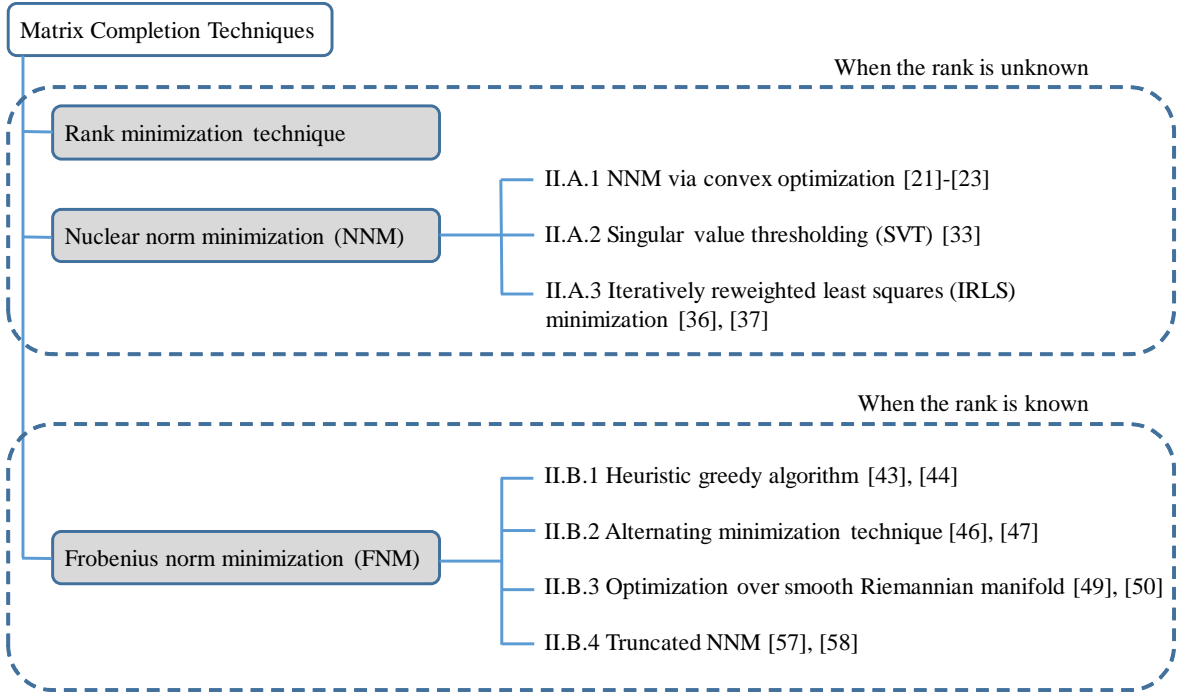} 
\caption {Outline of LRMC algorithms.} 
\end{figure*}

\item {\bf Image compression and restoration}: 
When there is dirt or scribble in a two-dimensional image (see Fig. \ref{fig:fig006}), one simple solution is to replace the contaminated pixels with the interpolated version of adjacent pixels. A better way is to exploit intrinsic domination of a few singular values in an image. In fact, one can readily approximate an image to the low-rank matrix without perceptible loss of quality. By using clean (uncontaminated) pixels as observed entries, an original image can be recovered via the low-rank matrix completion.  

\item {\bf Massive multiple-input multiple-output (MIMO)}: By exploiting hundreds of antennas at the basestation (BS), massive MIMO can offer a large gain in capacity. In order to optimize the performance gain of the massive MIMO systems, the channel state information at the transmitter (CSIT) is required~\cite{hyoungju2017}. One way to acquire the CSIT is to let each user directly feed back its own pilot observation to BS for the joint CSIT estimation of all users~\cite{shen2015}. In this setup, the MIMO channel matrix $\mathbf{H}$ can be reconstructed in two steps: 1) finding the pilot matrix $\mathbf{Y}$ using the least squares (LS) estimation or linear minimum mean square error (LMMSE) estimation and 2) reconstructing $\mathbf{H}$ using the model $\mathbf{Y} = \mathbf{H}\boldsymbol\Phi$ where each column of $\boldsymbol\Phi$ is the pilot signal from one antenna at BS \cite{marzetta2006, rusek2013}. Since the number of resolvable paths $P$ is limited in most cases, one can readily assume that $\text{rank}(\mathbf{H})\leq P$ \cite{shen2015}. In the massive MIMO systems, $P$ is often much smaller than the dimension of $\mathbf{H}$ due to the limited number of clusters around BS. Thus, the problem to recover $\mathbf{H}$ at BS can be solved via the rank minimization problem subject to the linear constraint $\mathbf{Y} = \mathbf{H}\boldsymbol\Phi$ \cite{rusek2013}.  

\end{enumerate}


Other than these, there are a bewildering variety of applications of LRMC in wireless communication, such as millimeter wave (mmWave) channel estimation \cite{rappaport2013,li2018}, topological interference management (TIM) \cite{shi2016,shi2017,shi2018,srid2015} and mobile edge caching in fog radio access networks (Fog-RAN)~\cite{peng2016,kyang2016}.

The paradigm of LRMC has received much attention ever since the works of Fazel \cite{fazel}, Candes and Recht \cite{candes_recht}, and Candes and Tao \cite{candes_tao}. Over the years, there have been lots of works on this topic~\cite{candes2015,truncatedNNM,mishra2014,dai2010}, but it might not be easy to grasp the essentials of LRMC from these studies. One reason is because many of these works are highly theoretical and based on random matrix theory, graph theory, manifold analysis, and convex optimization so that it is not easy to grasp the essential knowledge from these studies. Another reason is that most of these works are proposals of new LRMC technique so that it is difficult to catch a general idea and big picture of LRMC from these.

The primary goal of this paper is to provide a contemporary survey on LRMC, a new paradigm to recover unknown entries of a low-rank matrix from partial observations. To provide better view, insight, and understanding of the potentials and limitations of LRMC to researchers and practitioners in a friendly way, we present early scattered results in a structured and accessible way. Firstly, we classify the state-of-the-art LRMC techniques into two main categories and then explain each category in detail. Secondly, we present issues to be considered when using LRMC techniques. Specifically, we discuss the intrinsic properties required for low-rank matrix recovery and explain how to exploit a special structure, such as positive semidefinite-based structure, Euclidean distance-based structure, and graph structure, in LRMC design. Thirdly, we compare the recovery performance and the computational complexity of LRMC techniques via numerical simulations. We conclude the paper by commenting on the choice of LRMC techniques and providing future research directions.

Recently, there have been a few overview papers on LRMC. An overview of LRMC algorithms and their performance guarantees can be found in~\cite{davenport2016overview}. A survey with an emphasis on first-order LRMC techniques together with their computational efficiency is presented in~\cite{chen2018harnessing}. Our work is distinct from the previous studies in several aspects. Firstly, we categorize the state-of-the-art LRMC techniques into two classes and then explain the details of each class, which can help researchers to easily determine which technique can be used for the given problem setup. Secondly, we provide a comprehensive survey of LRMC techniques and also provide extensive simulation results on the recovery quality and the running time complexity from which one can easily see the pros and cons of each LRMC technique and also gain a better insight into the choice of LRMC algorithms. Finally, we discuss how to exploit a special structure of a low-rank matrix in the LRMC algorithm design. In particular, we introduce the CNN-based LRMC algorithm that exploits the graph structure of a low-rank matrix.


We briefly summarize notations used in this paper. 
\begin{itemize}
\item For a vector $\mathbf{a} \in \mathbb{R}^{n}$, $\text{diag}(\mathbf{a}) \in \mathbb{R}^{n \times n}$ is the diagonal matrix formed by $\mathbf{a}$.

\item For a matrix $\mathbf{A} \in \mathbb{R}^{n_{1} \times n_{2}}$, $\mathbf{a}_{i} \in \mathbb{R}^{n_{1}}$ is the $i$-th column of $\mathbf{A}$.

\item $\text{rank}(\mathbf{A})$ is the rank of $\mathbf{A}$.

\item $\mathbf{A}^{T} \in \mathbb{R}^{n_{2} \times n_{1}}$ is the transpose of $\mathbf{A}$.

\item For $\mathbf{A}, \mathbf{B} \in \mathbb{R}^{n_{1} \times n_{2}}$, $\langle \mathbf{A}, \mathbf{B} \rangle = \text{tr}(\mathbf{A}^{T} \mathbf{B})$ and $\mathbf{A} \odot \mathbf{B}$ are the inner product and the Hadamard product (or element-wise multiplication) of two matrices $\mathbf{A}$ and $\mathbf{B}$, respectively, where $\text{tr}(\cdot)$ denotes the trace operator.

\item $\| \mathbf{A} \|$, $\| \mathbf{A} \|_{\ast}$, and $\| \mathbf{A} \|_{F}$ stand for the spectral norm (i.e., the largest singular value), the nuclear norm (i.e., the sum of singular values), and the Frobenius norm of $\mathbf{A}$, respectively.

\item $\sigma_{i}(\mathbf{A})$ is the $i$-th largest singular value of $\mathbf{A}$.

\item $\mathbf{0}_{d_{1} \times d_{2}}$ and $\mathbf{1}_{d_{1} \times d_{2}}$ are $(d_{1} \times d_{2})$-dimensional matrices with entries being zero and one, respectively.

\item $\mathbf{I}_{d}$ is the $d$-dimensional identity matrix.

\item If $\mathbf{A}$ is a square matrix (i.e., $n_{1} = n_{2} = n$), $\text{diag}(\mathbf{A}) \in \mathbb{R}^{n}$ is the vector formed by the diagonal entries of $\mathbf{A}$.
\item $\text{vec}(\mathbf{X})$ is the vectorization of $\mathbf{X}$.
\end{itemize}


\section{Basics of Low-Rank Matrix Completion}

In this section, we discuss the principle to recover a low-rank matrix from partial observations. Basically, the desired low-rank matrix $\mathbf{M}$ can be recovered by solving the rank minimization problem
\begin{equation}
\begin{split}
    \min\limits_\mathbf{X} &~~~~~~~ \text{rank}(\mathbf{X}) \\
    \text{subject to} &~~~~~~~ x_{ij} = m_{ij},~ (i,j) \in {\Omega}, \label{eq:rankmin0}
\end{split}
\end{equation}
where $\Omega$ is the index set of observed entries (e.g., $\Omega = \{ (1,1), (1,2), (2,1) \}$ in the example in~\eqref{eq:eq001}). One can alternatively express the problem using the sampling operator $P_{\Omega}$. The sampling operation $P_{\Omega}(\mathbf{A})$ of a matrix $\mathbf{A}$ is defined as
$$[P_{\Omega} (\mathbf{A})]_{ij} 
= \begin{cases}
a_{ij} & \text{if } (i, j) \in \Omega \\
0 & \text{otherwise}.
\end{cases}$$
Using this operator, the problem~\eqref{eq:rankmin0} can be equivalently formulated as
\begin{equation} \label{eq:rankmin}
\begin{split}
\min\limits_\mathbf{X} &~~~~~~\text{rank}(\mathbf{X}) \\
\text{subject to}  &~~~~~~P_{\Omega}(\mathbf{X}) = P_{\Omega}(\mathbf{M}).
\end{split}
\end{equation}
A naive way to solve the rank minimization problem~\eqref{eq:rankmin} is the combinatorial search. Specifically, we first assume that $\text{rank}(\mathbf{M}) = 1$. Then, any two columns of $\textbf{M}$ are linearly dependent and thus we have the system of expressions $\mathbf{m}_i = \alpha_{i,j}\mathbf{m}_j$ for some $\alpha_{i,j}\in\mathbb{R}$. If the system has no solution for the rank-one assumption, then we move to the next assumption of $\text{rank}(\mathbf{M}) = 2$. In this case, we solve the new system of expressions $\mathbf{m}_i = \alpha_{i,j}\mathbf{m}_{j} + \alpha_{i,k}\mathbf{m}_{k}$. This procedure is repeated until the solution is found. Clearly, the combinatorial search strategy would not be feasible for most practical scenarios since it has an exponential complexity in the problem size~\cite{rankmin_complexity}. For example, when $\mathbf{M}$ is an $n \times n$ matrix, it can be shown that the number of the system expressions to be solved is $\mathcal{O}(n2^{n})$.

As a cost-effective alternative, various low-rank matrix completion (LRMC) algorithms have been proposed over the years. Roughly speaking, depending on the way of using the rank information, the LRMC algorithms can be classified into two main categories: 1) those without the
rank information and 2) those exploiting the rank information. In this section, we provide an in depth discussion of two categories (see the outline of LRMC algorithms in Fig.~\ref{fig:fig001}). 

\subsection{LRMC Algorithms Without the Rank Information}

In this subsection, we explain the LRMC algorithms that do not require the rank information of the original low-rank matrix.

\subsubsection{Nuclear Norm Minimization (NNM)}
\label{sec:nuclearnorm}

Since the rank minimization problem~\eqref{eq:rankmin} is NP-hard~\cite{fazel}, it is computationally intractable when the dimension of a matrix is large. One common trick to avoid computational issue is to replace the non-convex objective function with its convex surrogate, converting the combinatorial search problem into a convex optimization problem. There are two clear advantages in solving the convex optimization problem: 1) a local optimum solution is globally optimal and 2) there are many efficient polynomial-time convex optimization solvers (e.g., interior point method~\cite{boyd} and semi-definite programming (SDP) solver). 


In the LRMC problem, the nuclear norm $\| \mathbf{X} \|_{*}$, the sum of the singular values of $\mathbf{X}$, has been widely used as a convex surrogate of $\text{rank}(\mathbf{X})$~\cite{candes_recht}:
\begin{equation}\label{eq:nnm}
\begin{split}
\min\limits_\mathbf{X} &~~~~~~\| \mathbf{X} \|_{*} \\
\text{subject to}  &~~~~~~P_{\Omega} (\mathbf{X}) = P_{\Omega} (\mathbf{M})
\end{split}
\end{equation}  
Indeed, it has been shown that the nuclear norm is the convex envelope (the ``best'' convex approximation) of the rank function on the set $\{ \mathbf{X} \in \mathbb{R}^{n_{1} \times n_{2}} : \| \mathbf{X} \| \le 1 \}$ ~\cite{fazel}.\footnote{For any function $f:\mathcal{C} \rightarrow \mathbb{R}$, where $\mathcal{C}$ is a convex set, the convex envelope of $f$ is the largest convex function $g$ such that $f(x) \geq g(x)$ for all $x \in \mathcal{C}$. Note that the convex envelope of $\text{rank}(\mathbf{X})$ on the set $\{\mathbf{X} \in \mathbb{R}^{n_1 \times n_2}:\|\mathbf{X}\|\leq 1 \}$ is the nuclear norm $\|\mathbf{X}\|_*$ \cite{fazel}.} Note that the relaxation from the rank function to the nuclear norm is conceptually analogous to the relaxation from $\ell_0$-norm to $\ell_1$-norm in compressed sensing (CS) \cite{choi2017,kwon2014,wang2012}.

Now, a natural question one might ask is whether the NNM problem in (\ref{eq:nnm}) would offer a solution comparable to the solution of the rank minimization problem in~\eqref{eq:rankmin}. In~\cite{candes_recht}, it has been shown that if the observed entries of a rank $r$ matrix $\mathbf{M} (\in \mathbb{R}^{n \times n})$ are suitably random and the number of observed entries satisfies 
\begin{equation}
|\Omega| \ge C \mu_0n^{1.2}r \log n,
\label{eq:eq213}
\end{equation}
where $\mu_0$ is the largest coherence of $\mathbf{M}$ (see the definition in Subsection \ref{sec:coherence}), then $\mathbf{M}$ is the unique solution of the NNM problem \eqref{eq:nnm} with overwhelming probability (see Appendix \ref{app:appB}). 

%

It is worth mentioning that the NNM problem in~\eqref{eq:nnm} can also be recast as a semidefinite program (SDP) as (see Appendix \ref{app:appA})
\begin{equation} \label{eq:eqSDP_final}
\begin{split}
\min\limits_{\mathbf{Y}} &~~~~~ \text{tr}(\mathbf{Y}) \\
\text{subject to} &~~~~~ \langle \mathbf{Y}, \mathbf{A}_{k} \rangle = b_k,~~ k = 1, \cdots , |\Omega| \\
&~~~~~  \mathbf{Y} \succeq 0,
\end{split}
\end{equation}
where $\mathbf{Y} = \begin{bmatrix}
\mathbf{W}_{1} & \mathbf{X} \\
\mathbf{X}^{T} & \mathbf{W}_{2}
\end{bmatrix} \in \mathbb{R}^{(n_{1} + n_{2}) \times (n_{1} + n_{2})}$, $\{ \mathbf{A}_{k} \}_{k=1}^{|\Omega|}$ is the sequence of linear sampling matrices, and $\{ b_{k} \}_{k=1}^{|\Omega|}$ are the observed entries.
The problem \eqref{eq:eqSDP_final} can be solved by the off-the-shelf SDP solvers such as SDPT3~\cite{toh} and SeDuMi~\cite{sturm1999using} using interior-point methods \cite{van1996,zhang1998,nesterov1998,potra1998,van2005,potra2000}.  It has been shown that the computational complexity of SDP techniques is $\mathcal{O}(n^3)$ where $n = \max(n_1, n_2)$~\cite{van2005}. Also, it has been shown that under suitable conditions, the output $\widehat{\mathbf{M}}$ of SDP satisfies $\|\widehat{\mathbf{M}} - \mathbf{M}\|_F \leq \epsilon$ in at most $\mathcal{O}(n^{\omega}\log(\frac{1}{\epsilon}))$ iterations where $\omega$ is a positive constant~\cite{potra2000}. 
Alternatively, one can reconstruct $\mathbf{M}$ by solving the equivalent nonconvex quadratic optimization form of the NNM problem~\cite{recht2010}. Note that this approach has computational benefit since the number of primal variables of NNM is reduced from $n_1n_2$ to $r(n_1+n_2)$ ($r\leq \min(n_1,n_2)$). Interested readers may refer to~\cite{recht2010} for more details.
\subsubsection{Singular Value Thresholding (SVT)}

While the solution of the NNM problem in~\eqref{eq:nnm} can be obtained by solving~\eqref{eq:eqSDP_final}, this procedure is computationally burdensome when the size of the matrix is large. 

As an effort to mitigate the computational burden, the singular value thresholding (SVT) algorithm has been proposed~\cite{svt}. The key idea of this approach is to put the regularization term into the objective function of the NNM problem:
\begin{equation} \label{eq:svt1}
\begin{split}
    \min\limits_\mathbf{X} &~~~~~ \tau \| \mathbf{X} \|_{*} + \frac{1}{2} \| \mathbf{X} \|_{F}^2 \\
    \text{subject to} &~~~~~ P_{\Omega} (\mathbf{X}) = P_{\Omega} (\mathbf{M}),
\end{split}
\end{equation}
where $\tau$ is the regularization parameter. In \cite[Theorem 3.1]{svt}, it has been shown that the solution to the problem~\eqref{eq:svt1} converges to the solution of the NNM problem as $\tau \rightarrow \infty$.\footnote{In practice, a large value of $\tau$ has been suggested (e.g., $\tau=5n$ for an $n\times n$ low rank matrix) for the fast convergence of SVT. For example, when $\tau = 5000$, it requires 177 iterations to reconstruct a 1000$\times$1000 matrix of rank 10 \cite{svt}.}

Let $\mathcal{L}(\mathbf{X}, \mathbf{Y})$ be the Lagrangian function associated with~\eqref{eq:svt1}, i.e.,
\begin{equation}
\mathcal{L}(\mathbf{X},\mathbf{Y}) = \tau\|\mathbf{X}\|_\ast + \frac{1}{2}\|\mathbf{X}\|_F^2 + \langle \mathbf{Y},P_{\Omega}(\mathbf{M}) - P_{\Omega}(\mathbf{X}) \rangle
\end{equation}
where $\mathbf{Y}$ is the dual variable. Let $\widehat{\mathbf{X}}$ and $\widehat{\mathbf{Y}}$ be the primal and dual optimal solutions. Then, by the strong duality~\cite{boyd}, we have
\begin{equation}
\max\limits_\mathbf{Y}\min\limits_\mathbf{X}\mathcal{L}(\mathbf{X},\mathbf{Y}) = \mathcal{L}(\widehat{\mathbf{X}}, \widehat{\mathbf{Y}}) = \min\limits_\mathbf{X}\max\limits_\mathbf{Y}\mathcal{L}(\mathbf{X},\mathbf{Y}).
\label{eq:eqSVT004}
\end{equation}
The SVT algorithm finds $\widehat{\mathbf{X}}$ and $\widehat{\mathbf{Y}}$ in an iterative fashion. Specifically, starting with $\mathbf{Y}_{0} = \mathbf{0}_{n_{1} \times n_{2}}$, SVT updates $\mathbf{X}_{k}$ and $\mathbf{Y}_{k}$ as
\begin{subequations}
\begin{align}
&\mathbf{X}_{k} = \arg \underset{\mathbf{X}}{\min} \ \mathcal{L} ( \mathbf{X}, \mathbf{Y}_{k-1} ), \label{eq:svt_update equation of X_k_1} \\
&\mathbf{Y}_{k} = \mathbf{Y}_{k-1} + \delta_{k} \frac{\partial \mathcal{L}(\mathbf{X}_{k}, \mathbf{Y})}{\partial \mathbf{Y}}, \label{eq:svt_update equation of Y_k_1}
\end{align}
\end{subequations}
where $\{ \delta_{k} \}_{k \ge 1}$ is a sequence of positive step sizes. Note that $\mathbf{X}_{k}$ can be expressed as
\begin{align}
\mathbf{X}_{k}  
&\hspace{.37mm}=  \arg\min\limits_{\mathbf{X}} \tau\|\mathbf{X}\|_\ast + \frac{1}{2}\|\mathbf{X}\|_F^2 - \langle \mathbf{Y}_{k-1},P_{\Omega}(\mathbf{X}) \rangle \notag\\
&\overset{(a)}{=}  \arg\min\limits_{\mathbf{X}} \tau\|\mathbf{X}\|_\ast + \frac{1}{2}\|\mathbf{X}\|_F^2 - \langle P_{\Omega}(\mathbf{Y}_{k-1}),\mathbf{X} \rangle \notag\\
&\overset{(b)}{=} \arg\min\limits_{\mathbf{X}} \tau\|\mathbf{X}\|_\ast + \frac{1}{2}\|\mathbf{X}\|_F^2 - \langle \mathbf{Y}_{k-1},\mathbf{X} \rangle \nonumber \\
&\hspace{.37mm}=  \arg\min\limits_{\mathbf{X}} \tau\|\mathbf{X}\|_\ast + \frac{1}{2}\|\mathbf{X}-\mathbf{Y}_{k-1} \|_F^2, \label{eq:svt_update equation of X_k_2}
\end{align}
where (a) is because $\langle P_{\Omega}(\mathbf{A}),\mathbf{B} \rangle = \langle \mathbf{A},P_{\Omega}(\mathbf{B}) \rangle$ and (b) is because $\mathbf{Y}_{k-1}$ vanishes outside of $\Omega$ (i.e., $P_{\Omega}(\mathbf{Y}_{k-1}) = \mathbf{Y}_{k-1}$) by~\eqref{eq:svt_update equation of Y_k_1}. Due to the inclusion of the nuclear norm, finding out the solution $\mathbf{X}_{k}$ of~\eqref{eq:svt_update equation of X_k_2} seems to be difficult. However, thanks to the intriguing result of Cai et al., we can easily obtain the solution. 
\begin{theorem}[\hspace{-.2mm}{\cite[Theorem 2.1]{svt}}] \label{thm:thmSVT001}
Let $\mathbf{Z}$ be a matrix whose singular value decomposition (SVD) is $\mathbf{Z} = \mathbf{U} \mathbf{\Sigma} \mathbf{V}^T$. Define $t_{+} = \max \{ t, 0 \}$ for $t \in \mathbb{R}$. Then,
\begin{equation}
\mathcal{D}_{\tau}(\mathbf{Z}) 
= \arg\min\limits_{\mathbf{X}}\tau\|\mathbf{X}\|_\ast + \frac{1}{2}\|\mathbf{X}-\mathbf{Z}\|_F^2,
\label{eq:eq003}
\end{equation} 
where $\mathcal{D}_{\tau}$ is the singular value thresholding operator defined as
\begin{align}
\mathcal{D}_{\tau} (\mathbf{Z}) 
&= \mathbf{U} \diag ( \{ ( \sigma_{i}(\boldsymbol\Sigma) - \tau )_{+} \}_{i} \} ) \mathbf{V}^{T}.
\label{eq:eq201}
\end{align}
\end{theorem} 
By Theorem~\ref{thm:thmSVT001}, the right-hand side of~\eqref{eq:svt_update equation of X_k_2} is $\mathcal{D}_{\tau}(\mathbf{Y}_{k-1})$. To conclude, the update equations for $\mathbf{X}_{k}$ and $\mathbf{Y}_{k}$ are given by
\begin{subequations}
\begin{align}
&\mathbf{X}_{k} = \mathcal{D}_{\tau}(\mathbf{Y}_{k-1}), \label{eq:svt_update equation of X_k} \\
&\mathbf{Y}_{k} = \mathbf{Y}_{k-1} + \delta_{k} (P_{\Omega}(\mathbf{M}) - P_{\Omega}(\mathbf{X}_{k})). \label{eq:svt_update equation of Y_k}
\end{align}
\end{subequations}
One can notice from \eqref{eq:svt_update equation of X_k} and \eqref{eq:svt_update equation of Y_k} that the SVT algorithm is computationally efficient since we only need the truncated SVD and elementary matrix operations in each iteration.
Indeed, let $r_{k}$ be the number of singular values of $\mathbf{Y}_{k-1}$ being greater than the threshold $\tau$. Also, we suppose $\{r_{k}\}$ converges to the rank of the original matrix, i.e., $\lim_{k\rightarrow\infty}r_k = r$. Then the computational complexity of SVT is $\mathcal{O}(rn_1n_2)$. Note also that the iteration number to achieve the $\epsilon$-approximation\footnote{By $\epsilon$-approximation, we mean $\|\widehat{\mathbf{M}}-\mathbf{M}^\ast\|_F\leq \epsilon$ where $\widehat{\mathbf{M}}$ is the reconstructed matrix and $\mathbf{M}^\ast$ is the optimal solution of SVT.} is $\mathcal{O}(\frac{1}{\sqrt{\epsilon}})$~\cite{svt}.
In Table~\ref{tab:svt}, we summarize the SVT algorithm. For the details of the stopping criterion of SVT, see \cite[Section 5]{svt}.

Over the years, various SVT-based techniques have been proposed~\cite{tanner, combettes, jain}. 
In \cite{combettes}, an iterative matrix completion algorithm using the SVT-based operator called proximal operator has been proposed. Similar algorithms inspired by the iterative hard thresholding (IHT) algorithm in CS have also been proposed~\cite{tanner, jain}.


\setlength{\arrayrulewidth}{1pt}
\begin{table}
  \centering
\caption{The SVT Algorithm} \label{tab:svt}
\vspace{-2mm}
\begin{tabular}{@{}ll}
\hline \vspace{-7pt} \\
\textbf{Input}       & observed entries $P_\Omega(\mathbf{M})$, \\
					 & a sequence of positive step sizes $ \{ \delta_{k} \}_{k \ge 1}$, \\
					 & a regularization parameter $\tau > 0$, \\
					 & and a stopping criterion $T$ \\ 		
\textbf{Initialize}  & iteration counter $k = 0$ \\
					 & and $\mathbf{Y}_{0} = \mathbf{0}_{n_{1} \times n_{2}}$, \\
\hline \vspace{-7pt} \\
\textbf{While}		 & $T = \text{false}$~~\textbf{do} \\
					 & $k = k + 1$ \\
					 & $[\mathbf{U}_{k-1}, \mathbf{\Sigma}_{k-1}, \mathbf{V}_{k-1}] = \text{svd}(\mathbf{Y}_{k-1})$ \\					                      
					 & $\mathbf{X}_{k} = \mathbf{U}_{k-1} \diag ( \{ ( \sigma_{i}(\boldsymbol\Sigma_{k-1}) - \tau )_{+} \}_{i} \} ) \mathbf{V}_{k-1}^{T}$ using \eqref{eq:eq201} \\                     
					 & $\mathbf{Y}_{k} = \mathbf{Y}_{k-1} + \delta_k(P_{\Omega}(\mathbf{M})-P_{\Omega}(\mathbf{X}_{k}))$ \\              
\textbf{End} \\
\hline \vspace{-7pt} \\					       
\textbf{Output}  & $\mathbf{X}^{k}$ \\
\vspace{-7pt} \\
\hline
\end{tabular}
\end{table} 


\subsubsection{Iteratively Reweighted Least Squares (IRLS) Minimization}
Yet another simple and computationally efficient way to solve the NNM problem is the IRLS minimization technique \cite{fornasier2011,mohan2012}. In essence, the NNM problem can be recast using the least squares minimization as
\begin{equation}
\begin{split}
\min\limits_{\mathbf{X},\mathbf{W}} &~~~~~~\|\mathbf{W}^{\frac{1}{2}}\mathbf{X}\|_F^2 \\
\text{subject to}  &~~~~~~P_{\Omega} (\mathbf{X}) = P_{\Omega} (\mathbf{M}),
\end{split}
\label{eq:eq212}
\end{equation} 
where $\mathbf{W} = (\mathbf{XX}^T)^{-\frac{1}{2}}$. It can be shown that \eqref{eq:eq212} is equivalent to the NNM problem \eqref{eq:nnm} since we have \cite{fornasier2011}
\begin{equation}
\|\mathbf{X}\|_\ast  = tr((\mathbf{XX}^T)^{\frac{1}{2}}) = \|\mathbf{W}^{\frac{1}{2}}\mathbf{X}\|_F^2.
\end{equation}
The key idea of the IRLS technique is to find $\mathbf{X}$ and $\mathbf{W}$ in an iterative fashion. The update expressions are
\begin{subequations}
\begin{align}
\label{eq:eq203}
\mathbf{X}_{k} & =  \text{arg}\min\limits_{P_\Omega(\mathbf{X}) = P_\Omega(\mathbf{M})}\|\mathbf{W}_{k-1}^{\frac{1}{2}}\mathbf{X}\|_F^2,\\
\label{eq:eq204}
\mathbf{W}_{k} & =  (\mathbf{X}_{k}\mathbf{X}_{k}^T)^{-\frac{1}{2}}.
\end{align}  
\end{subequations}
Note that the weighted least squares subproblem \eqref{eq:eq203} can be easily solved by updating each and every column of $\mathbf{X}_k$ \cite{fornasier2011}.  
In order to compute $\mathbf{W}_k$, we need a matrix inversion \eqref{eq:eq204}. To avoid the ill-behavior (i.e., some of the singular values of $\mathbf{X}_k$ approach to zero), an approach to use the perturbation of singular values has been proposed \cite{fornasier2011,mohan2012}.
Similar to SVT, the computational complexity per iteration of the IRLS-based technique is $\mathcal{O}(rn_1n_2)$. 
Also, IRLS requires $\mathcal{O}(\log(\frac{1}{\epsilon}))$ iterations to achieve the $\epsilon$-approximation solution. 
We summarize the IRLS minimization technique in Table \ref{tab:irls}.



\setlength{\arrayrulewidth}{1pt}
\begin{table}
  \centering
\caption{The IRLS Algorithm} \label{tab:irls}
\vspace{-2mm}

\begin{tabular}{@{}ll}
\hline \vspace{-7pt} \\
\textbf{Input}       & a constant $q \geq r$, \\
					 & a scaling parameter $\gamma > 0$, \\
					 & and a stopping criterion $T$ \\ 		
\textbf{Initialize}  & iteration counter $k = 0$, \\
					 & a regularizing sequence $\epsilon_0 = 1$,\\					 
					 & and $\mathbf{W}_{0} = \mathbf{I}$ \\
\hline \vspace{-7pt} \\
\textbf{While}		 & $T = \text{false}$~~\textbf{do} \\
					 & $k = k + 1$ \\
					 & $\mathbf{X}_{k} = \text{arg}\min_{P_\Omega(\mathbf{X}) = P_\Omega(\mathbf{M})}\|\mathbf{W}_{k-1}^{\frac{1}{2}}\mathbf{X}\|_F^2$ \\                     
					 & $\epsilon_k = \min(\epsilon_{k-1}, \gamma\sigma_{q+1}(\mathbf{X}_k))$ \\                     
					 & Compute a SVD perturbation version $\widetilde{\mathbf{X}}_k$ of $\mathbf{X}_k$ \cite{fornasier2011}\\					 
					 & $\mathbf{W}_{k} = (\widetilde{\mathbf{X}}_{k}\widetilde{\mathbf{X}}_k^T)^{-\frac{1}{2}}$ \\              
\textbf{End} \\
\hline \vspace{-7pt} \\					       
\textbf{Output}  & $\mathbf{X}^{k}$ \\
\vspace{-7pt} \\
\hline
\end{tabular}

\end{table}

\subsection{LRMC Algorithms Using Rank Information}

In many applications such as localization in IoT networks, recommendation system, and image restoration, we encounter the situation where the rank of a desired matrix is known in advance. As mentioned, the rank of a Euclidean distance matrix in a localization problem is at most $k+2$ ($k$ is the dimension of the Euclidean space). In this situation, the LRMC problem can be formulated as a Frobenius norm minimization (FNM) problem:
\begin{equation} \label{eq:lsm}
\begin{split}
    \min\limits_\mathbf{X} &~~~~~ \frac{1}{2}\| P_{\Omega} (\mathbf{M}) - P_{\Omega} (\mathbf{X}) \|_{F}^2 \\
    \text{subject to} &~~~~~ \text{rank} (\mathbf{X}) \leq r. 
\end{split}
\end{equation}
Due to the inequality of the rank constraint, an approach to use the approximate rank information (e.g., upper bound of the rank) has been proposed~\cite{adm}.
The FNM problem has two main advantages: 1) the problem is well-posed in the noisy scenario and 2) the cost function is differentiable so that various gradient-based optimization techniques (e.g., gradient descent, conjugate gradient, Newton methods, and manifold optimization) can be used
to solve the problem.

Over the years, various techniques to solve the FNM problem in~\eqref{eq:lsm} have been proposed~\cite{adm,rom,pf,asd,lmafit,mishra2014,dai2010,bart2013,sgg,truncatedNNM}. 
The performance guarantee of the FNM-based techniques has also been provided~\cite{ge2016,ge2017,du2017}. 
It has been shown that under suitable conditions of the sampling ratio $p = |\Omega|/(n_1n_2)$ and the largest coherence $\mu_0$ of $\mathbf{M}$ (see the definition in Subsection \ref{sec:coherence}), the gradient-based algorithms globally converges to $\mathbf{M}$ with high probability~\cite{ge2017}.  
Well-known FNM-based LRMC techniques include greedy techniques \cite{adm}, alternating projection techniques \cite{pf}, and optimization over Riemannian manifold \cite{bart2013}. In this subsection, we explain these techniques in detail.

\subsubsection{Greedy Techniques}
\label{sec:greedy}
In recent years, greedy algorithms have been popularly used for LRMC due to the computational simplicity. In a nutshell, they solve the LRMC problem by making a heuristic decision at each iteration with a hope to find the right solution in the end.


Let $r$ be the rank of a desired low-rank matrix $\mathbf{M} \in \mathbb{R}^{n \times n}$ and $\mathbf{M} = \mathbf{U} \mathbf{\Sigma} \mathbf{V}^{T}$ be the singular value decomposition of $\mathbf{M}$ where $\mathbf{U}, \mathbf{V} \in \mathbb{R}^{n \times r}$. By noting that 
\begin{align}
\mathbf{M} 
&= \sum\limits_{i = 1}^{r}\sigma_i(\mathbf{M}) \mathbf{u}_i\mathbf{v}_i^T,
\end{align} 
$\mathbf{M}$ can be expressed as a linear combination of $r$ rank-one matrices.
The main task of greedy techniques is to investigate the \textit{atom set} $\mathcal{A}_{\mathbf{M}} = \{ \boldsymbol{\varphi}_{i} = \mathbf{u}_{i} \mathbf{v}_{i}^{T} \}_{i=1}^{r}$ of rank-one matrices representing $\mathbf{M}$. Once the atom set $\mathcal{A}_{\mathbf{M}}$ is found, the singular values $\sigma_{i}(\mathbf{M}) = \sigma_{i}$ can be computed easily by solving the following problem
\begin{align}
(\sigma_1,\cdots,\sigma_r) 
&= \arg\min\limits_{\alpha_i}\|P_{\Omega}(\mathbf{M}) - P_{\Omega}(\sum\limits_{i=1}^{r}\alpha_i\boldsymbol\varphi_i)\|_F.
\label{eq:eq005}
\end{align}
To be specific, let $\mathbf{A} = [\text{vec} (P_{\Omega} (\boldsymbol{\varphi}_{1})) \ \cdots \ \text{vec}(P_{\Omega}(\boldsymbol{\varphi}_{r}))]$, $\boldsymbol{\alpha} = [\alpha_{1} \ \cdots \ \alpha_{r}]^{T}$ and $\mathbf{b} = \text{vec}(P_{\Omega}(\mathbf{M}))$. Then, we have 
$(\sigma_1,\cdots,\sigma_r)  = \text{arg}\min\limits_{\boldsymbol\alpha}\|\mathbf{b} - \mathbf{A}\boldsymbol\alpha\|_2 = \mathbf{A}^{\dagger}\mathbf{b}$.


One popular greedy technique is atomic decomposition for minimum rank approximation (ADMiRA)~\cite{adm}, which can be viewed as an extension of the compressive sampling matching pursuit (CoSaMP) algorithm in CS~\cite{needell2009,choi2017,kwon2014,wang2012}. ADMiRA employs a strategy of adding as well as pruning to identify the atom set $\mathcal{A}_{\mathbf{M}}$. In the adding stage, ADMiRA identifies $2r$ rank-one matrices representing a residual best and then adds the matrices to the pre-chosen atom set. Specifically, if $\mathbf{X}_{i-1}$ is the output matrix generated in the $(i-1)$-th iteration and $\mathcal{A}_{i-1}$ is its atom set, then ADMiRA computes the residual $\mathbf{R}_{i} = P_{\Omega}(\mathbf{M}) - P_{\Omega}(\mathbf{X}_{i-1})$ and then adds $2r$ leading principal components of $\mathbf{R}_{i}$ to $\mathcal{A}_{i-1}$. In other words, the enlarged atom set $\Psi_{i}$ is given by 
\begin{align} \label{eq:enlarged atom set_ADMiRA}
\Psi_{i}
&= \mathcal{A}_{i-1} \cup \{ \mathbf{u}_{\mathbf{R}_{i}, j} \mathbf{v}_{\mathbf{R}_{i}, j}^{T} : 1 \le j \le 2r \},
\end{align}
where $\mathbf{u}_{\mathbf{R}_{i}, j}$ and $\mathbf{v}_{\mathbf{R}_{i}, j}$ are the $j$-th principal left and right singular vectors of $\mathbf{R}_{i}$, respectively. Note that $\Psi_{i}$ contains at most $3r$ elements. In the pruning stage, ADMiRA refines $\Psi_{i}$ into a set of $r$ atoms. To be specific, if $\widetilde{\mathbf{X}}_{i}$ is the best rank-$3r$ approximation of $\mathbf{M}$, i.e.,\footnote{Note that the solution to (\ref{eq:eq006}) can be computed in a similar way as in (\ref{eq:eq005}).}
\begin{align}
\widetilde{\mathbf{X}}_{i} 
&= \arg\min\limits_{\mathbf{X}\in\text{span}(\Psi_{i})}\|P_{\Omega}(\mathbf{M}) - P_{\Omega}(\mathbf{X})\|_F,
\label{eq:eq006}
\end{align}
then the refined atom set $\mathcal{A}_{i}$ is expressed as
\begin{align} \label{eq:updated atom set_ADMiRA}
\mathcal{A}_{i}
&= \{ \mathbf{u}_{\widetilde{\mathbf{X}}_{i}, j} \mathbf{v}_{\widetilde{\mathbf{X}}_{i}, j}^{T} : 1 \le j \le r \},
\end{align}
where $\mathbf{u}_{\widetilde{\mathbf{X}}_{i}, j}$ and $\mathbf{v}_{\widetilde{\mathbf{X}}_{i}, j}$ are the $j$-th principal left and right singular vectors of $\widetilde{\mathbf{X}}_{i}$, respectively. 
The computational complexity of ADMiRA is mainly due to two operations: the least squares operation in \eqref{eq:eq005} and the SVD-based operation to find out the leading atoms of the required matrix (e.g., $\mathbf{R}_k$ and $\mathbf{\widetilde{\mathbf{X}}_{k+1}}$). First, since \eqref{eq:eq005} involves the pseudo-inverse of $\mathbf{A}$ (size of $|\Omega| \times \mathcal{O}(r)$), its computational cost is $\mathcal{O}(r|\Omega|)$. Second, the computational cost of performing a truncated SVD of $\mathcal{O}(r)$ atoms is $\mathcal{O}(rn_1n_2)$. Since $|\Omega| < n_1n_2$, the computational complexity of ADMiRA per iteration is $\mathcal{O}(rn_1n_2)$.
Also, the iteration number of ADMiRA to achieve the $\epsilon$-approximation is $\mathcal{O}(\log(\frac{1}{\epsilon}))$~\cite{adm}. 
In Table~\ref{tab:ADMiRA}, we summarize the ADMiRA algorithm.

\setlength{\arrayrulewidth}{1pt}
\begin{table}
  \centering
\caption{The ADMiRA Algorithm} \label{tab:ADMiRA}
\vspace{-2mm}
\begin{tabular}{@{}ll}
\hline \vspace{-7pt} \\
\textbf{Input}       & observed entries $P_\Omega(\mathbf{M}) \in \mathbb{R}^{n \times n}$, \\
					 & rank of a desired low-rank matrix $r$, \\
					 & and a stopping criterion $T$ \\
\textbf{Initialize}  & iteration counter $k = 0$, \\
					 & $\mathbf{X}_{0} = \mathbf{0}_{n \times n}$, \\
					 & and $\mathcal{A}_{0} = \emptyset$ \\
\hline \vspace{-7pt} \\
\textbf{While}		 & $T = \text{false}$~~\textbf{do} \\
					 & $\mathbf{R}_{k} = P_{\Omega}(\mathbf{M}) - P_{\Omega}(\mathbf{X}_{k})$ \\
					 & $[\mathbf{U}_{\mathbf{R}_{k}}, \mathbf{\Sigma}_{\mathbf{R}_{k}}, \mathbf{V}_{\mathbf{R}_{k}}] = \text{svds}(\mathbf{R}_{k}, 2r)$ \\
					 & (Augment)~~$\Psi_{k+1} = \mathcal{A}_{k} \cup \{ \mathbf{u}_{\mathbf{R}_{k}, j} \mathbf{v}_{\mathbf{R}_{k}, j}^{T} : 1 \le j \le 2r \}$ \\
					 & $\widetilde{\mathbf{X}}_{k+1} = \underset{\mathbf{X}\in\text{span}(\Psi_{k+1})}{\arg \min} \|P_{\Omega}(\mathbf{M}) - P_{\Omega}(\mathbf{X})\|_F$ using \eqref{eq:eq005} \\ 
					 & $[\mathbf{U}_{\widetilde{\mathbf{X}}_{k+1}}, \mathbf{\Sigma}_{\widetilde{\mathbf{X}}_{k+1}}, \mathbf{V}_{\widetilde{\mathbf{X}}_{k+1}}] = \text{svds}({\widetilde{\mathbf{X}}_{k+1}}, r)$ \\                   
					 & (Prune)~~$\mathcal{A}_{k+1} = \{ \mathbf{u}_{\widetilde{\mathbf{X}}_{k+1}, j} \mathbf{v}_{\widetilde{\mathbf{X}}_{k+1}, j}^{T} : 1 \le j \le r \}$  \\
					 & (Estimate)~~$\mathbf{X}_{k+1} = \underset{\mathbf{X}\in\text{span}(\mathcal{A}_{k+1})}{\arg \min} \|P_{\Omega}(\mathbf{M}) - P_{\Omega}(\mathbf{X})\|_F$   \\
					 & \qquad\qquad\quad using \eqref{eq:eq005} \\
					 & $k = k+1$ \\     
\textbf{End} \\
\hline \vspace{-7pt} \\	                
\textbf{Output}  &$\mathcal{A}_{k}$, $\mathbf{X}_{k}$ \\
\vspace{-7pt} \\
\hline
\end{tabular}
\end{table}

Yet another well-known greedy method is the rank-one matrix pursuit algorithm \cite{rom}, an extension of the orthogonal matching pursuit algorithm in CS \cite{tropp2007}. In this approach, instead of choosing multiple atoms of a matrix, an atom corresponding to the largest singular value of the residual matrix $\mathbf{R}_k$ is chosen. 



\vspace{.25cm}
\subsubsection{Alternating Minimization Techniques}
%
Many of LRMC algorithms \cite{svt,adm} require the computation of (partial) SVD to obtain the singular values and vectors (expressed as $\mathcal{O}(rn^2)$). As an effort to further reduce the computational burden of SVD, alternating minimization techniques have been proposed \cite{pf,asd, lmafit}. The basic premise behind the alternating minimization techniques is that a low-rank matrix $\mathbf{M} \in \mathbb{R}^{n_{1} \times n_{2}}$ of rank $r$ can be factorized into tall and fat matrices, i.e., $\mathbf{M} = \mathbf{X} \mathbf{Y}$ where $\mathbf{X}\in\mathbb{R}^{n_1\times r}$ and $ \mathbf{Y}\in \mathbb{R}^{r\times n_2}$ $(r \ll n_{1}, n_{2})$. The key idea of this approach is to find out $\mathbf{X}$ and $\mathbf{Y}$ minimizing the residual defined as the difference between the original matrix and the estimate of it on the sampling space. In other words, they recover $\mathbf{X}$ and $\mathbf{Y}$ by solving \begin{align} \label{eq:alternative minimization techniques}
\underset{\mathbf{X}, \mathbf{Y}}{\min}~~~~~\frac{1}{2} \| P_{\Omega}(\mathbf{M}) - P_{\Omega}(\mathbf{X} \mathbf{Y}) \|_{F}^{2}.
\end{align}
Power factorization, a simple alternating minimization algorithm, finds out the solution to~\eqref{eq:alternative minimization techniques} by updating $\mathbf{X}$ and $\mathbf{Y}$ alternately as~\cite{pf}
\begin{subequations}
\begin{align}
&\mathbf{X}_{i+1} = \arg\min_\mathbf{X} \|P_\Omega(\mathbf{M}) - P_\Omega(\mathbf{X}\mathbf{Y}_i) \|_{F}^{2},  \\
&\mathbf{Y}_{i+1} = \arg\min_\mathbf{Y} \|P_\Omega(\mathbf{M}) - P_\Omega(\mathbf{X}_{i+1}\mathbf{Y}) \|_{F}^{2}. 
\end{align}
\end{subequations}
Alternating steepest descent (ASD) is another alternating method to find out the solution~\cite{asd}. The key idea of ASD is to update $\mathbf{X}$ and $\mathbf{Y}$ by applying the steepest gradient descent method to the objective function $f(\mathbf{X},\mathbf{Y})= \frac{1}{2} \| P_{\Omega}(\mathbf{M}) - P_{\Omega}(\mathbf{X} \mathbf{Y}) \|_{F}^{2}$ in~\eqref{eq:alternative minimization techniques}. Specifically, ASD first computes the gradient of $f(\mathbf{X}, \mathbf{Y})$ with respect to $\mathbf{X}$ and then updates $\mathbf{X}$ along the steepest gradient descent direction: 
\begin{align}
\mathbf{X}_{i+1} 
&= \mathbf{X}_i-t_{x_i} \bigtriangledown f_{\mathbf{Y}_{i}}(\mathbf{X}_{i}),
\end{align}
where the gradient descent direction $\bigtriangledown f_{\mathbf{Y}_{i}}(\mathbf{X}_{i})$ and stepsize $t_{x_{i}}$ are given by
\begin{subequations}
\begin{align}
&\bigtriangledown f_{\mathbf{Y}_{i}}(\mathbf{X}_{i}) = -(P_\Omega(\mathbf{M}) - P_\Omega(\mathbf{X}_{i} \mathbf{Y}_{i}))\mathbf{Y}_{i}^T, \\
&t_{x_i} = \frac{\| \bigtriangledown f_{\mathbf{Y}_{i}}(\mathbf{X}_{i}) \|_F^2}{\|P_\Omega(\bigtriangledown f_{\mathbf{Y}_i}(\mathbf{X}_i)\mathbf{Y}_i)\|_F^2}.
\end{align}
\end{subequations}
After updating $\mathbf{X}$, ASD updates $\mathbf{Y}$ in a similar way:
\begin{align}
\mathbf{Y}_{i+1} 
&= \mathbf{Y}_{i}-t_{y_i}\bigtriangledown f_{\mathbf{X}_{i+1}}(\mathbf{Y}_{i}),
\end{align}
where 
\begin{subequations}
\begin{align}
&\bigtriangledown f_{\mathbf{X}_{i+1}}(\mathbf{Y}_{i}) = -\mathbf{X}_{i+1}^T(P_\Omega(\mathbf{M}) - P_\Omega(\mathbf{X}_{i+1}\mathbf{Y}_i)), \\
&t_{y_i} = \frac{\|\bigtriangledown f_{\mathbf{X}_{i+1}}(\mathbf{Y}_i)\|_F^2}{\|P_\Omega(\mathbf{X}_{i+1}\bigtriangledown f_{\mathbf{X}_{i+1}}(\mathbf{Y}_i))\|_F^2}.
\end{align}
\end{subequations}

The low-rank matrix fitting (LMaFit) algorithm finds out the solution in a different way by solving~\cite{lmafit} 
\begin{equation}
\arg\min_{\mathbf{X},\mathbf{Y},\mathbf{Z}} \{ \|\mathbf{X}\mathbf{Y}-\mathbf{Z}\|_F^2 : P_{\Omega}(\mathbf{Z}) = P_{\Omega}(\mathbf{M}) \}.
\label{AL}
\end{equation} 
With the arbitrary input of $\mathbf{X}_0\in \mathbb{R}^{n_1\times r}$ and $\mathbf{Y}_0\in \mathbb{R}^{r\times n_2}$ and $\mathbf{Z}_0 = P_\Omega(\mathbf{M})$, the variables $\mathbf{X}$, $\mathbf{Y}$, and $\mathbf{Z}$ are updated in the $i$-th iteration as
\begin{subequations}
\begin{align}
&\mathbf{X}_{i+1} = \arg\min_\mathbf{X}\|\mathbf{X}\mathbf{Y}_i-\mathbf{Z}_i\|_F^2 = \mathbf{Z}_i\mathbf{Y}^\dagger,  \\
&\mathbf{Y}_{i+1} = \arg\min_\mathbf{Y}\|\mathbf{X}_i\mathbf{Y}-\mathbf{Z}_i\|_F^2 = \mathbf{X}_{i+1}^\dagger \mathbf{Z}_i,  \\
&\mathbf{Z}_{i+1} = \mathbf{X}_{i+1}\mathbf{Y}_{i+1}+P_\Omega (\mathbf{M}- \mathbf{X}_{i+1}\mathbf{Y}_{i+1}),  
\end{align}
\end{subequations}
where $\mathbf{X}^\dagger$ is Moore-Penrose pseudoinverse of matrix $\mathbf{X}$.

Running time of the alternating minimization algorithms is very short due to the following reasons: 1) it does not require the SVD computation and 2) the size of matrices to be inverted is smaller than the size of matrices for the greedy algorithms. While the inversion of huge size matrices (size of $|\Omega|\times \mathcal{O}(1)$) is required in a greedy algorithms (see \eqref{eq:eq005}), alternating minimization only requires the pseudo inversion of $\mathbf{X}$ and $\mathbf{Y}$ (size of $n_1\times r$ and $r\times n_2$, respectively). Indeed, the computational complexity of this approach is $\mathcal{O}(r|\Omega| + r^2n_1 + r^2n_2)$, which is much smaller than that of SVT and ADMiRA when $r\ll \min(n_1, n_2)$. Also, the iteration number of ASD and LMaFit to achieve the $\epsilon$-approximation is $\mathcal{O}(\log(\frac{1}{\epsilon}))$~\cite{asd,lmafit}. It has been shown that alternating minimization techniques are simple to implement and also require small sized memory \cite{apm}. Major drawback of these approaches is that it might converge to the local optimum.

\vspace{.25cm}
\subsubsection{Optimization over Smooth Riemannian Manifold}
%
In many applications where the rank of a matrix is known in a priori (i.e., $\text{rank}(\mathbf{M}) = r$), one can strengthen the constraint of (\ref{eq:lsm}) by defining the feasible set, denoted by $\mathcal{F}$, as $$\mathcal{F} = \{\mathbf{X}\in\mathbb{R}^{n_1\times n_2}\: :\: \text{rank}(\mathbf{X}) = r \}.$$ Note that $\mathcal{F}$ is not a vector space\footnote{This is because if $\text{rank}(\mathbf{X})= r$ and $\text{rank}(\mathbf{Y})= r$, then $\text{rank}(\mathbf{X} + \mathbf{Y})= r$ is not necessarily true (and thus $\mathbf{X}+\mathbf{Y}$ does not need to belong $\mathcal{F}$).} and thus conventional optimization techniques cannot be used to solve the problem defined over $\mathcal{F}$. While this is bad news, a remedy for this is that $\mathcal{F}$ is a smooth Riemannian manifold~\cite{helmke1994,mishra2014}. Roughly speaking, smooth manifold is a generalization of $\mathbb{R}^{n_1\times n_2}$ on which a notion of differentiability exists. For more rigorous definition, see, e.g.,~\cite{absil2009optimization, lee2003smooth}. A smooth manifold equipped with an inner product, often called a Riemannian metric, forms a smooth Riemannian manifold. Since the smooth Riemannian manifold is a differentiable structure equipped with an inner product, one can use all necessary ingredients to solve the optimization problem with quadratic cost function, such as Riemannian gradient, Hessian matrix, exponential map, and parallel translation \cite{absil2009optimization}. Therefore, optimization techniques in $\mathbb{R}^{n_1\times n_2}$ (e.g., steepest descent, Newton method, conjugate gradient method) can be used to solve (\ref{eq:lsm}) in the smooth Riemannian manifold $\mathcal{F}$.  

In recent years, many efforts have been made to solve the matrix completion over smooth Riemannian manifolds. These works are classified by their specific choice of Riemannian manifold structure. One well-known approach is to solve (\ref{eq:lsm}) over the Grassmann manifold of orthogonal matrices\footnote{The Grassmann manifold is defined as the set of the linear subspaces in a vector space \cite{absil2009optimization}.} \cite{dai2010}. In this approach, a feasible set can be expressed as $\mathcal{F} = \{\mathbf{QR}^T:\mathbf{Q}^T\mathbf{Q} = \mathbf{I},\mathbf{Q}\in\mathbb{R}^{n_1\times r},\mathbf{R}\in\mathbb{R}^{n_2\times r}\}$ and thus solving (\ref{eq:lsm}) is to find an $n_1\times r$ orthonormal matrix $\mathbf{Q}$ satisfying
\begin{equation}
f(\mathbf{Q}) = \min\limits_{\mathbf{R}\in\mathbb{R}^{n_2\times r}} \|P_{\Omega}(\mathbf{M})-P_{\Omega}(\mathbf{QR}^T)\|_F^2 = 0.
\label{eq:eq104}
\end{equation} 
In~\cite{dai2010}, an approach to solve \eqref{eq:eq104} over the Grassmann manifold has been proposed. 

Recently, it has been shown that the original matrix can be reconstructed by the unconstrained optimization over the smooth Riemannian manifold $\mathcal{F}$ \cite{bart2013}. Often, $\mathcal{F}$ is expressed using the singular value decomposition as
\begin{align} \label{eq:eqRM001}
\mathcal{F} 
&= \{\mathbf{U}\boldsymbol\Sigma\mathbf{V}^T: \mathbf{U}\in\mathbb{R}^{n_1\times r},\mathbf{V}\in\mathbb{R}^{n_2\times r},\boldsymbol\Sigma\succeq 0, \nonumber\\
&~~~~~~ \mathbf{U}^T\mathbf{U} = \mathbf{V}^T\mathbf{V} = \mathbf{I}, \boldsymbol\Sigma = \text{diag}([\sigma_1 \ \cdots \ \sigma_r]) \}.
\end{align}  
The FNM problem (\ref{eq:lsm}) can then be reformulated as an unconstrained optimization over $\mathcal{F}$:
\begin{align} 
\underset{\mathbf{X}\in\mathcal{F}}{\min}~~~~~\frac{1}{2} \| P_{\Omega}(\mathbf{M}) - P_{\Omega}(\mathbf{X})\|_{F}^{2}.
\end{align}
One can easily obtain the closed-form expression of the ingredients such as tangent spaces, Riemannian metric, Riemannian gradient, and Hessian matrix in the unconstrained optimization~\cite{helmke1994, absil2009optimization, lee2003smooth}. In fact, major benefits of the Riemannian optimization-based LRMC techniques are the simplicity in implementation and the fast convergence. 
Similar to ASD, the computational complexity per iteration of these techniques is $\mathcal{O}(r|\Omega|+r^2n_1+r^2n_2)$, and they require $\mathcal{O}(\log(\frac{1}{\epsilon}))$ iterations to achieve the $\epsilon$-approximation solution~\cite{bart2013}.

%
\subsubsection{Truncated NNM}
%
Truncated NNM is a variation of the NNM-based technique requiring the rank information $r$.\footnote{Although truncated NNM is a variant of NNM, we put it into the second category since it exploits the rank information of a low-rank matrix.} While the NNM technique takes into account all the singular values of a desired matrix, truncated NNM considers only the $n-r$ smallest singular values \cite{truncatedNNM}. Specifically, truncated NNM finds a solution to
\begin{equation} \label{eq:original truncated NNM formulation}
\begin{split}
    \min\limits_\mathbf{X} &~~~~~ \| \mathbf{X} \|_{r} \\
    \text{subject to} &~~~~~ P_{\Omega} (\mathbf{X}) = P_{\Omega} (\mathbf{M}), 
\end{split}    
\end{equation}
where $\|\mathbf{X}\|_r = \sum\limits_{i=r+1}^{n}\sigma_{i}(\mathbf{X})$. We recall that $\sigma_{i}(\mathbf{X})$ is the $i$-th largest singular value of $\mathbf{X}$. Using~\cite{truncatedNNM}
\begin{align}
\underset{i=1}{\overset{r}{\sum}} \sigma_{i}
&= \underset{\mathbf{U}^{T}\mathbf{U} = \mathbf{V}^{T}\mathbf{V} = \mathbf{I}_{r}}{\max} \text{tr}(\mathbf{U}^{T} \mathbf{X} \mathbf{V}),
\end{align}
we have
\begin{align} \label{eq:alternate expression of truncated nuclear norm}
\| \mathbf{X} \|_{r}
&= \| \mathbf{X} \|_{\ast} - \underset{\mathbf{U}^{T}\mathbf{U} = \mathbf{V}^{T}\mathbf{V} = \mathbf{I}_{r}}{\max} \text{tr}(\mathbf{U}^{T} \mathbf{X} \mathbf{V}),
\end{align}
and thus the problem \eqref{eq:original truncated NNM formulation} can be reformulated to 
\begin{equation} \label{eq:alternate form of truncated NNM formulation}
\begin{split}
    \min\limits_\mathbf{X} &~~~~~ \| \mathbf{X} \|_{\ast} - \underset{\mathbf{U}^{T}\mathbf{U} = \mathbf{V}^{T}\mathbf{V} = \mathbf{I}_{r}}{\max} \text{tr}(\mathbf{U}^{T} \mathbf{X} \mathbf{V}) \\
    \text{subject to} &~~~~~ P_{\Omega} (\mathbf{X}) = P_{\Omega} (\mathbf{M}), 
\end{split}    
\end{equation}
This problem can be solved in an iterative way. Specifically, starting from $\mathbf{X}_{0} = P_{\Omega}(\mathbf{M})$, truncated NNM updates $\mathbf{X}_{i}$ by solving~\cite{truncatedNNM}
\begin{align} \label{eq:iterative way to solve the truncated NNM problem}
\begin{split}
    \min\limits_\mathbf{X} &~~~~~ \|\mathbf{X}\|_\ast - \text{tr}(\mathbf{U}_{i-1}^T\mathbf{X}\mathbf{V}_{i-1}) \\
    \text{subject to} &~~~~~ P_{\Omega} (\mathbf{X}) = P_{\Omega} (\mathbf{M}),
\end{split}
\end{align}
where $\mathbf{U}_{i-1}, \mathbf{V}_{i-1} \in \mathbb{R}^{n \times r}$ are the matrices of left and right-singular vectors of $\mathbf{X}_{i-1}$, respectively.
We note that an approach in \eqref{eq:iterative way to solve the truncated NNM problem} has two main advantages: 1) the rank information of the desired matrix can be incorporated and 2) various gradient-based techniques including alternating direction method of multipliers (ADMM)~\cite{tao2011,lin2011}, ADMM with an adaptive penalty (ADMMAP) \cite{he2000}, and accelerated proximal gradient line search method (APGL)\cite{beck2009} can be employed. Note also that the dominant operation is the truncated SVD operation and its complexity is $\mathcal{O}(rn_1n_2)$, which is much smaller than that of the NNM technique (see Table \ref{tab:tab001}) as long as $r\ll\min(n_1,n_2)$.
Similar to SVT, the iteration complexity of the truncated NNM to achieve the $\epsilon$-approximation is $\mathcal{O}(\frac{1}{\sqrt{\epsilon}})$~\cite{truncatedNNM}.
Alternatively, the difference of two convex functions (DC) based algorithm can be used to solve \eqref{eq:alternate form of truncated NNM formulation}~\cite{DCA}. In Table~\ref{tab:truncated NNM}, we summarize the truncated NNM algorithm.

\setlength{\arrayrulewidth}{1pt}
\begin{table}
  \centering
\caption{Truncated NNM} \label{tab:truncated NNM}
\vspace{-2mm}
\begin{tabular}{@{}ll}
\hline \vspace{-7pt} \\
\textbf{Input}       & observed entries $P_\Omega(\mathbf{M}) \in \mathbb{R}^{n \times n}$, \\
					 & rank of a desired low-rank matrix $r$, \\
					 & and stopping threshold $\epsilon > 0$\\
\textbf{Initialize}  & iteration counter $k = 0$, \\
					 & and $\mathbf{X}_{0} = P_{\Omega}(\mathbf{M})$ \\
\hline \vspace{-7pt} \\
\textbf{While}		 & $\| \mathbf{X}_{k} - \mathbf{X}_{k-1} \|_{F} > \epsilon$~~\textbf{do} \\
					 & $[\mathbf{U}_{k}, \mathbf{\Sigma}_{k}, \mathbf{V}_{k}] = \text{svd}(\mathbf{X}_{k})$ ($\mathbf{U}_{k}, \mathbf{V}_{k} \in \mathbb{R}^{n \times r}$) \\
					 & $\mathbf{X}^{k+1} = \underset{\mathbf{X}: P_{\Omega}(\mathbf{X}) = P_{\Omega}(\mathbf{M})}{\arg \min} \| \mathbf{X} \|_{\ast} - \text{tr}(\mathbf{U}_{k}^{T} \mathbf{X} \mathbf{V}_{k})$ \\                     
					 & $k = k+1$ \\     
\textbf{End} \\
\hline \vspace{-7pt} \\	                
\textbf{Output}  &$\mathbf{X}^{k}$ \\
\vspace{-7pt} \\
\hline
\end{tabular}
\end{table}

\section{Issues to be Considered When Using LRMC Techniques}

In this section, we study the main principles that make the recovery of a low-rank matrix possible and discuss how to exploit a special structure of a low-rank matrix in algorithm design.

\subsection{Intrinsic Properties}
There are two key properties characterizing the LRMC problem: 1) \textit{sparsity} of the observed entries and 2) \textit{incoherence} of the matrix. Sparsity indicates that an accurate recovery of the undersampled matrix is possible even when the number of observed entries is very small. Incoherence indicates that nonzero entries of the matrix should be spread out widely for the efficient recovery of a low-rank matrix. 
In this subsection, we go over these issues in detail.

\subsubsection{Sparsity of Observed Entries}
Sparsity expresses an idea that when a matrix has a low rank property, then it can be recovered using only a small number of observed entries. Natural question arising from this is how many elements do we need to observe for the accurate recovery of the matrix.
In order to answer this question, we need to know a notion of a degree of freedom (DOF). The DOF of a matrix is the number of freely chosen variables in the matrix.
One can easily see that the DOF of the rank one matrix in (\ref{eq:eq001}) is $3$ since one entry can be determined after observing three.
As an another example, consider the following rank one matrix 
\be
    \mathbf{M} = \matc{ 1 & 3 & 5 & 7 \\ 2 & 6 & 10 & 14 \\ 3 & 9 & 15 & 21 \\ 4 & 12 & 20 & 28 }.
\ee
%
One can easily see that if we observe all entries of one column and one row, then the rest can be determined by a simple linear relationship between these since $\mathbf{M}$ is the rank-one matrix. Specifically, if we observe the first row and the first column, then the first and the second columns differ by the factor of three so that as long as we know one entry in the second column, rest will be recovered. Thus, the DOF of $\mathbf{M}$ is $4 + 4 - 1 = 7$.
Following lemma generalizes our observations.

\begin{lemma} \label{thm:dof} The DOF of a square $n\times n$ matrix with rank $r$ is $2nr - r^2$. Also, the DOF of $n_1 \times n_2$-matrix is $(n_1 + n_2) r - r^2$.
\end{lemma}

\begin{proof}
Since the rank of a matrix is $r$,  we can freely choose values for all entries of the $r$ columns, resulting in $nr$ degrees of freedom for the first $r$ column.
Once $r$ independent columns, say $\mathbf{m}_1, \cdots \mathbf{m}_r$, are constructed, then each of the rest $n-r$ columns is expressed as a linear combinations of the first $r$ columns (e.g., $\mathbf{m}_{r+1} = \alpha_1 \mathbf{m}_1 + \cdots + \alpha_r \mathbf{m}_r$) so that $r$ linear coefficients ($\alpha_1, \cdots \alpha_r$) can be freely chosen in these columns.
By adding $nr$ and $(n-r)r$, we obtain the desired result. Generalization to $n_1\times n_2$ matrix is straightforward.
\end{proof}



\begin{figure*} [t!]
\label{fig:mc_illustration}
\centering
\subfigure[] 
{\includegraphics[width=0.45\textwidth]{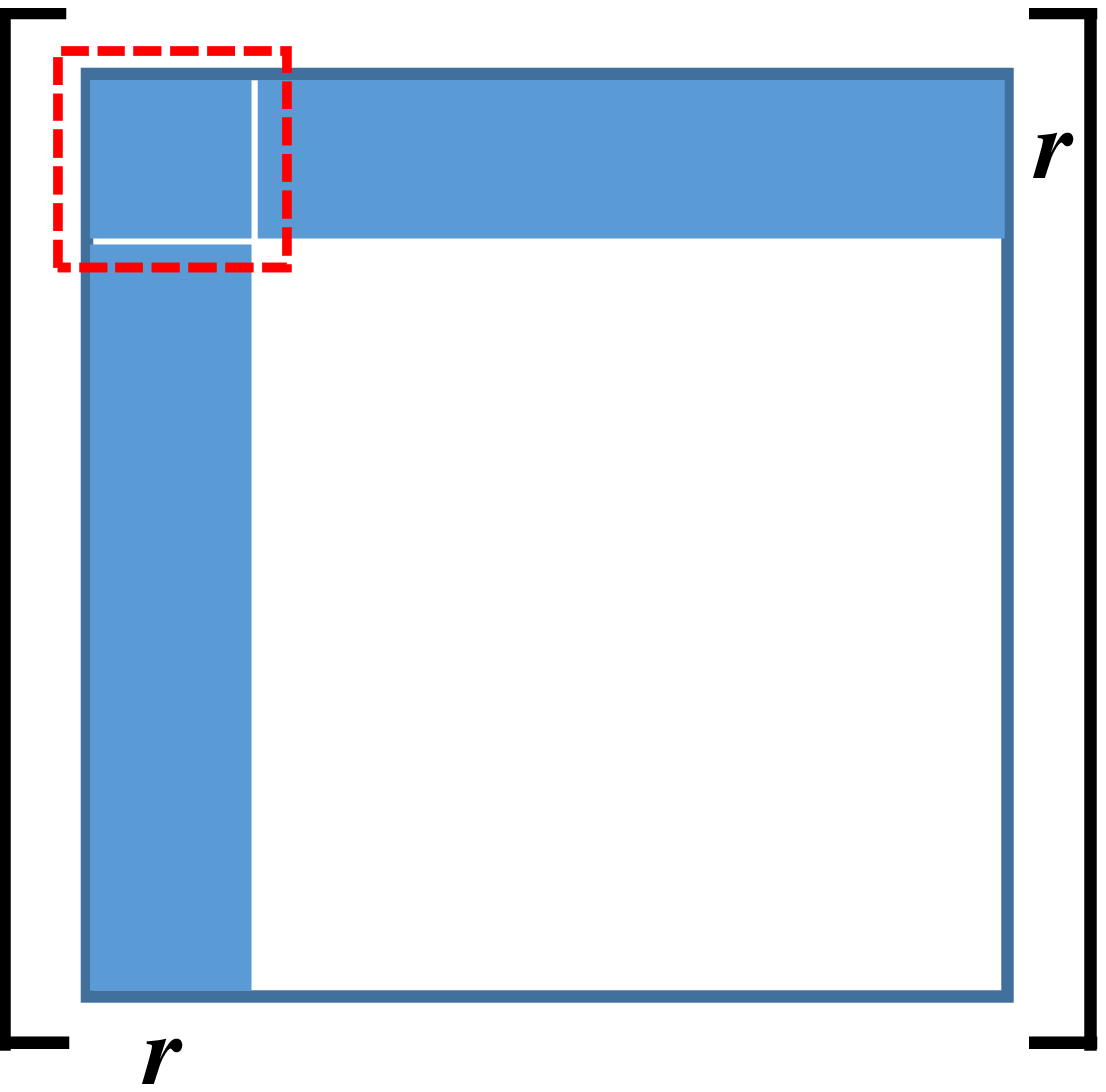} } 
\subfigure[] 
{\includegraphics[width=0.45\textwidth]{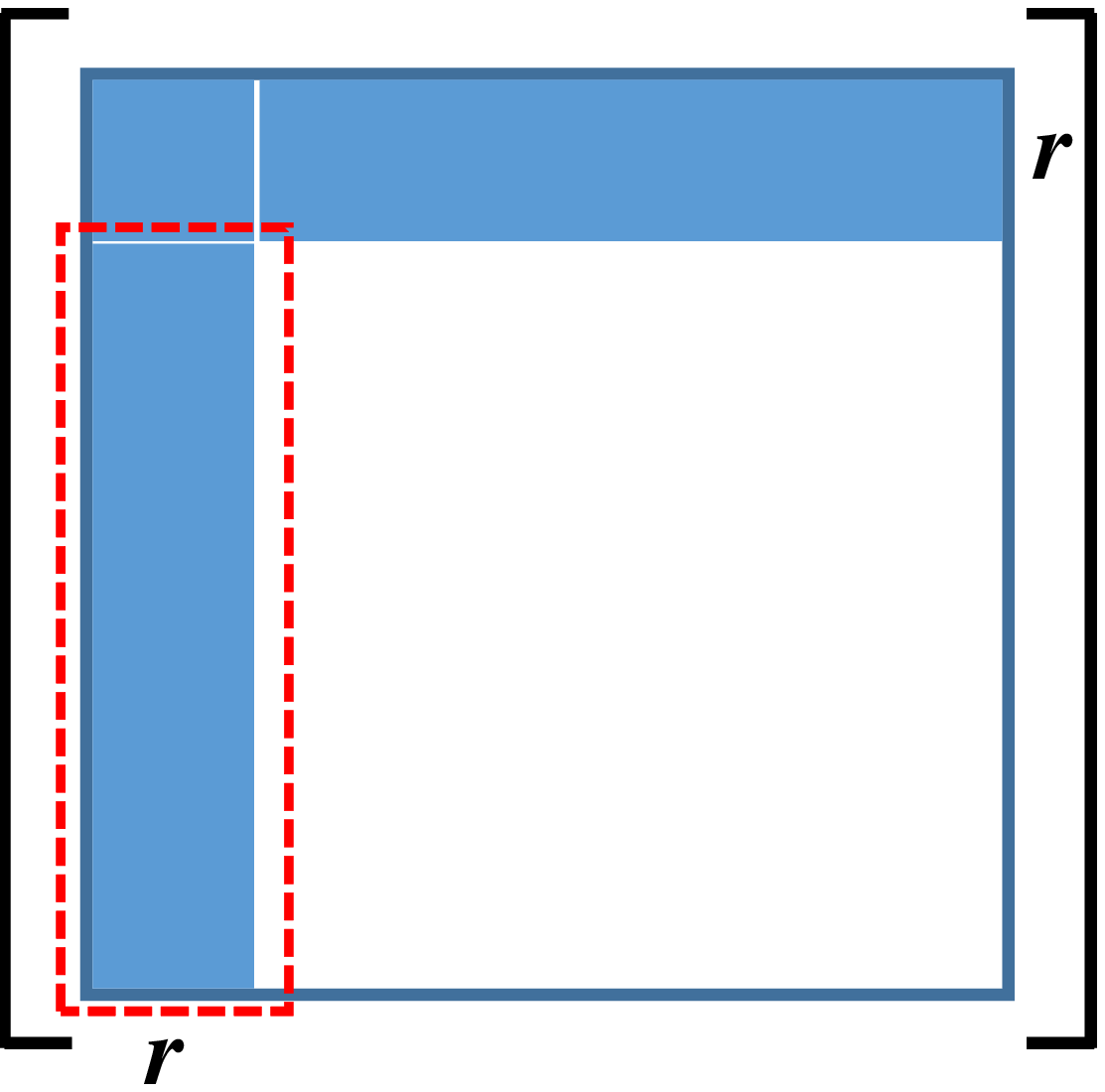} }   
\caption {LRMC with colored entries being observed. The dotted boxes are used to compute: (a) linear coefficients and (b) unknown entries.} 
\end{figure*}


This lemma says that if $n$ is large and $r$ is small enough (e.g., $r = O(1)$), essential information in a matrix is just in the order of $n$, DOF$=O(n)$, which is clearly much smaller than the total number of entries of the matrix. 
Interestingly, the DOF is the minimum number of observed entries required for the recovery of a matrix. If this condition is violated, that is, if the number of observed entries is less than the DOF (i.e., $m < 2nr - r^2$), no algorithm whatsoever can recover the matrix.
In Fig. \ref{fig:mc_illustration}, we illustrate how to recover the matrix when the number of observed entries equals the DOF. In this figure, we assume that blue colored entries are observed.\footnote{Since we observe the first $r$ rows and columns, we have $2nr-r^2$ observations in total.}
In a nutshell, unknown entries of the matrix are found in two-step process. First, we identify the linear relationship between the first $r$ columns and the rest. For example, the $(r+1)$-th column can be expressed as a linear combination of the first $r$ columns. That is,
\begin{align} \label{eq:linear}
\mathbf{m}_{r+1} 
&= \alpha_1 \mathbf{m}_1 + \cdots + \alpha_r \mathbf{m}_r.
\end{align}
Since the first $r$ entries of $\mathbf{m}_1, \cdots \mathbf{m}_{r+1}$ are observed (see Fig. \ref{fig:mc_illustration}(a)), we have $r$ unknowns $(\alpha_1, \cdots, \alpha_r)$ and $r$ equations so that we can identify the linear coefficients $\alpha_1, \cdots \alpha_r$ with the computational cost $\mathcal{O}(r^3)$ of an $r\times r$ matrix inversion. Once these coefficients are identified, we can recover the unknown entries  $m_{r+1,r+1} \cdots m_{r+1,n}$ of $\mathbf{m}_{r+1}$ using the linear relationship in (\ref{eq:linear}) (see Fig. \ref{fig:mc_illustration}(b)).
By repeating this step for the rest of columns, we can identify all unknown entries with $\mathcal{O}(rn^2)$ computational complexity\footnote{For each unknown entry, it needs $r$ multiplication and $r - 1$ addition operations. Since the number of unknown entries is $(n-r)^2$, the computational cost is $(2r-1)(n-r)^2$. Recall that $\mathcal{O}(r^3)$ is the cost of computing $(\alpha_1, \cdots, \alpha_r)$ in \eqref{eq:linear}. Thus, the total cost is $\mathcal{O}(r^3 + (2r-1)(n-r)^2) = \mathcal{O}(rn^2)$.}.


\begin{figure*} [t!]
\label{fig:mc_illustration4}
\centering    
\includegraphics[width=0.4\textwidth]{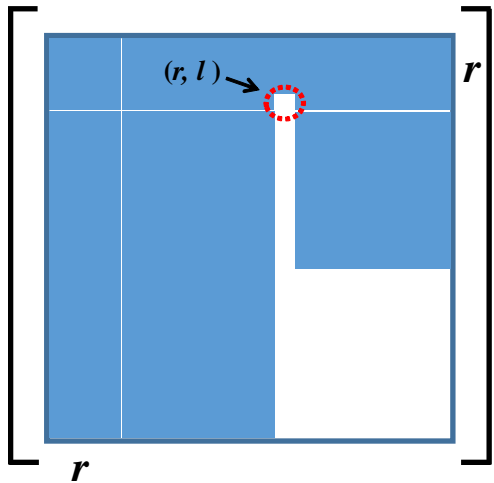} 
\caption {An illustration of the worst case of LRMC.} 
\end{figure*}


Now, an astute reader might notice that this strategy will not work if one entry of the column (or row) is unobserved. As illustrated in Fig. \ref{fig:mc_illustration4}, if only one entry in the $r$-th row, say $(r,l)$-th entry, is unobserved, then one cannot recover the $l$-th column simply because the matrix in Fig. \ref{fig:mc_illustration4} cannot be converted to the matrix form in Fig. \ref{fig:mc_illustration}(b).
It is clear from this discussion that the measurement size being equal to the DOF is not enough for the most cases and in fact it is just a necessary condition for the accurate recovery of the rank-$r$ matrix. This seems like a depressing news. However, DOF is in any case important since it is a fundamental limit (lower bound) of the number of observed entries to ensure the exact recovery of the matrix. Recent results show that the DOF is not much different from the number of measurements ensuring the recovery of the matrix~\cite{candes_recht,recht}.\footnote{In \cite{recht}, it has been shown that the required number of entries to recover the matrix using the nuclear-norm minimization is in the order of $n^{1.2}$ when the rank is $O(1)$.} 

\subsubsection{Coherence}
\label{sec:coherence}


\begin{figure*} [t!]
\label{fig:coherence}
\centering  
\subfigure[]  
{\includegraphics[width=0.4\textwidth]{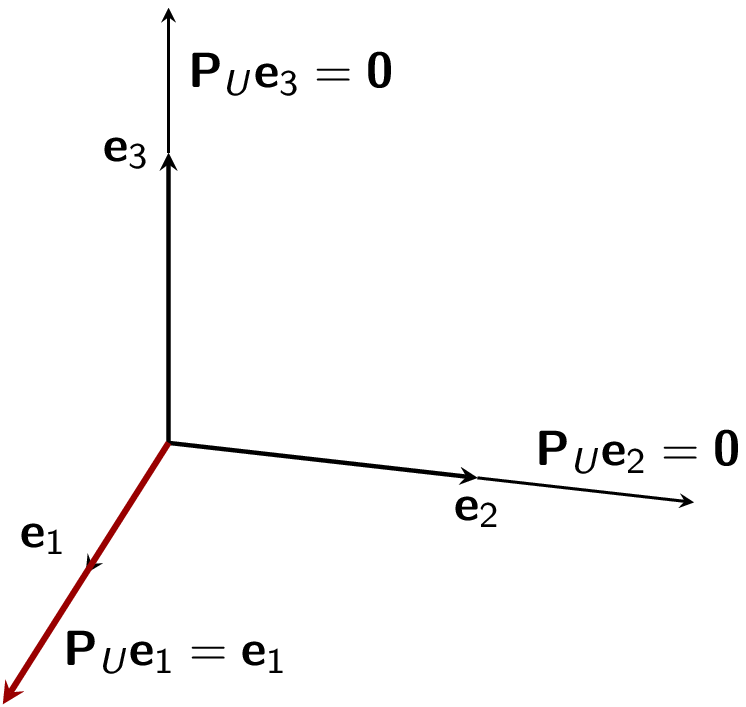} }  
\subfigure[]  
{\includegraphics[width=0.5\textwidth]{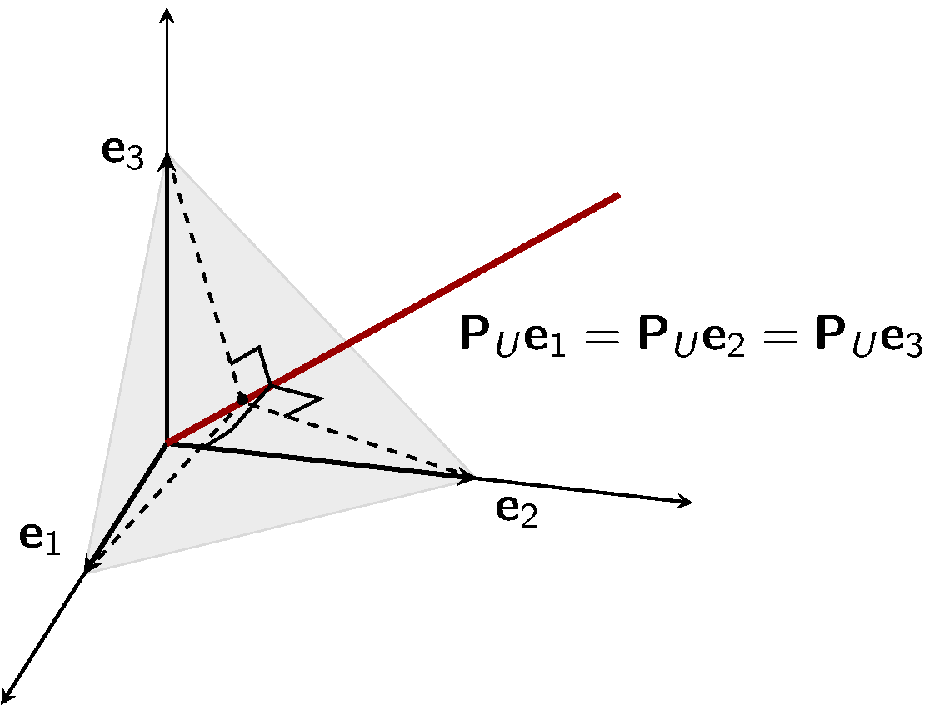} } 
\caption {Coherence of matrices in (\ref{eq:example1}) and (\ref{eq:example2}): (a) maximum and (b) minimum.} 
\end{figure*}

If nonzero elements of a matrix are concentrated in a certain region, we generally need a large number of observations to recover the matrix. On the other hand, if the matrix is spread out widely, then the matrix can be recovered with a relatively small number of entries. For example, consider the following two rank-one matrices in $\mathbb{R}^{n\times n}$
        \be
        \mathbf{M}_1 =  \matc{1 &  1 & 0 & \cdots & 0 \\
                            1 &  1 & 0 & \cdots & 0 \\
                            0 &  0 & 0 & \cdots & 0 \\
                            \vdots & \vdots & \vdots & \ddots & \vdots  \\
                            0 &  0 & 0 & \cdots & 0},\nonumber
        \ee  
        \be
        \mathbf{M}_2 = \matc{1 &  1 & 1 & \cdots & 1 \\
                            1 &  1 & 1 & \cdots & 1 \\
                            1 &  1 & 1 & \cdots & 1 \\
                            \vdots & \vdots & \vdots & \ddots & \vdots  \\
                            1 &  1 & 1 & \cdots & 1}.\nonumber.
        \ee      
        The matrix $\mathbf{M}_1$ has only four nonzero entries at the top-left corner. Suppose $n$ is large, say $n = 1000$, and all entries but the four elements in the top-left corner are observed (99.99\% of entries are known). In this case, even though the rank of a matrix is just one, there is no way to recover this matrix since the information bearing entries is missing. This tells us that although the rank of a matrix is very small, one might not recover it if nonzero entries of the matrix are concentrated in a certain area. 
        
        In contrast to the matrix $\mathbf{M}_1$, one can accurately recover the matrix $\mathbf{M}_2$ with only $2n-1$ (= DOF) known entries. In other words, one row and one column are enough to recover $\mathbf{M}_2$). One can deduce from this example that the spread of observed entries is important for the identification of unknown entries.
        
        In order to quantify this, we need to measure the concentration of a matrix. Since the matrix has two-dimensional structure, we need to check the concentration in both row and column directions. This can be done by checking the concentration in the left and right singular vectors. Recall that the SVD of a matrix is
        \be
            \mathbf{M} = \mathbf{U} \mathbf{\Sigma} \mathbf{V}^T = \sum_{i=1}^{r} \sigma_i \mathbf{u}_i \mathbf{v}_i^T
            \label{eq:svd}
        \ee
        where $\mathbf{U} = [\mathbf{u}_1 \ \cdots \ \mathbf{u}_r]$ and $\mathbf{V} = [\mathbf{v}_1 \ \cdots \ \mathbf{v}_r]$ are the matrices constructed by the left and right singular vectors, respectively, and $\mathbf{\Sigma}$ is the diagonal matrix whose diagonal entries are $\sigma_i$.
        From (\ref{eq:svd}), we see that the concentration on the vertical direction (concentration in the row) is determined by $\mathbf{u}_i$ and that on the horizontal direction (concentration in the column) is determined by $\mathbf{v}_i$.
        For example, if one of the standard basis vector $\mathbf{e}_i$, say $\mathbf{e}_1 = [1 ~0 \cdots ~0]^T$, lies on the space spanned by $\mathbf{u}_1, \cdots \mathbf{u}_r$ while others ($\mathbf{e}_2, \mathbf{e}_3, \cdots$) are orthogonal to this space, then it is clear that nonzero entries of the matrix are only on the first row.
        In this case, clearly one cannot infer the entries of the first row from the sampling of the other row. That is, it is not possible to recover the matrix without observing the entire entries of the first row.

        The coherence, a measure of concentration in a matrix, is formally defined as \cite{recht}
        \be
            \mu(\mathbf{U}) = \frac{n}{r} \max\limits_{1\leq i \leq n} \|P_\mathbf{U}\mathbf{e}_i\|^2
            \label{eq:coherence}
        \ee
        where $\mathbf{e}_i$ is standard basis and $P_\mathbf{U}$ is the projection onto the range space of $\mathbf{U}$. Since the columns of $\mathbf{U} = [\mathbf{u}_1 ~ \cdots ~ \mathbf{u}_r ]$ are orthonormal, we have $$P_\mathbf{U} = \mathbf{U} \mathbf{U}^{\dagger} = \mathbf{U}(\mathbf{U}^T\mathbf{U})^{-1}\mathbf{U}^T = \mathbf{U}\mathbf{U}^T.$$ 
        Note that both $\mu(\mathbf{U})$ and $\mu(\mathbf{V})$ should be computed to check the concentration on the vertical and horizontal directions.

\begin{lemma} (Maximum and minimum value of $\mu(\mathbf{U})$) \label{thm:max_min}
 $\mu(\mathbf{U})$ satisfies
\begin{equation}
  1 \leq \mu(\mathbf{U}) \leq \frac{n}{r}
  \label{eq:bound1}
\end{equation}
\end{lemma}

\begin{proof}
The upper bound is established by noting that $\ell_2$-norm of the projection is not greater than the original vector ($\| P_\mathbf{U} \mathbf{e}_i \|_2^2 \leq \|  \mathbf{e}_i \|_2^2$). The lower bound is because
\begin{align*}
    \max_i \| P_\mathbf{U} \mathbf{e}_i \|_2^2 
    &\ge \frac{1}{n} \sum_{i=1}^{n} \| P_\mathbf{U} \mathbf{e}_i \|_2^2 \nonumber  \\
    &= \frac{1}{n} \sum_{i=1}^{n} \mathbf{e}_i^T P_\mathbf{U} \mathbf{e}_i \nonumber \\
    &= \frac{1}{n} \sum_{i=1}^{n} \mathbf{e}_i^T \mathbf{U} \mathbf{U}^T \mathbf{e}_i \nonumber \\
    &= \frac{1}{n} \sum_{i=1}^{n} \sum_{j=1}^{r} | u_{ij} |^2 \nonumber \\
    &= \frac{r}{n} \nonumber
\end{align*}
where the first equality is due to the idempotency of $P_\mathbf{U}$ (i.e., $P_\mathbf{U}^T P_\mathbf{U} = P_\mathbf{U}$) and the last equality is because $\sum_{i=1}^{n} | u_{ij} |^2 = 1$.
\end{proof}

Coherence is maximized when the nonzero entries of a matrix are concentrated in a row (or column). For example, consider the matrix whose nonzero entries are concentrated on the first row
\begin{align}
\mathbf{M} 
&= \matc{3 & 2 & 1\\ 0 & 0 & 0\\ 0 & 0 & 0}. \label{eq:example1}
\end{align}
Note that the SVD of $\mathbf{M}$ is
\begin{align*}
    \mathbf{M} 
    &= \sigma_1\mathbf{u}_1\mathbf{v}_1^T = 3.8417 \matc{1 \\ 0 \\ 0} [0.8018 \ 0.5345 \ 0.2673].
\end{align*}
Then, $\mathbf{U} = [1 \ 0 \ 0]^{T}$, and thus $\| P_\mathbf{U} \mathbf{e}_1 \|_2 = 1$ and $\| P_\mathbf{U} \mathbf{e}_2 \|_2 = \| P_\mathbf{U} \mathbf{e}_3 \|_2 = 0$. As shown in Fig. \ref{fig:coherence}(a), the standard basis $\mathbf{e}_1$ lies on the space spanned by $\mathbf{U}$ while others are orthogonal to this space so that the maximum coherence is achieved ($\max_i \| P_\mathbf{U} \mathbf{e}_i \|_2^2 = 1$ and $\mu(\mathbf{U}) = 3$).
%


In contrast, coherence is minimized when the nonzero entries of a matrix are spread out widely. Consider the matrix
\begin{align}
\mathbf{M} 
&= \matc{2 & 1 & 0\\ 2 & 1 & 0\\ 2 & 1 & 0}.
\label{eq:example2}
\end{align}
In this case, the SVD of $\mathbf{M}$ is
\begin{align}
\mathbf{M} 
&= 3.8730 \matc{-0.5774 \\ -0.5774 \\ -0.5774 } [-0.8944 \ -0.4472 \ 0]. \nonumber
\end{align}
Then, we have
\begin{align*}
P_\mathbf{U} 
&= \mathbf{U} \mathbf{U}^{T} = \frac{1}{3}\matc{1 & 1 & 1\\ 1 & 1 & 1\\ 1 & 1 & 1},
\end{align*} 
and thus $\| P_\mathbf{U} \mathbf{e}_1 \|_2^2 = \| P_\mathbf{U} \mathbf{e}_2 \|_2 = \| P_\mathbf{U} \mathbf{e}_3 \|_2 = \frac{1}{3}$.
In this case, as illustrated in Fig. \ref{fig:coherence}(b), $\| P_\mathbf{U} \mathbf{e}_i \|_2$ is the same for all standard basis vector $\mathbf{e}_i$, achieving lower bound in (\ref{eq:bound1}) and the minimum coherence ($\max_i \| P_\mathbf{U} \mathbf{e}_i \|_2^2 = \frac{1}{3}$ and $\mu(\mathbf{U}) = 1$).
As discussed in \eqref{eq:eq213}, the number of measurements to recover the low-rank matrix is proportional to the coherence of the matrix~\cite{candes_recht,candes_tao,recht}.

\vspace{.25cm}
\subsection{Working With Different Types of Low-Rank Matrices}
\label{sec:sec001}

In many practical situations where the matrix has a certain structure, we want to make the most of the given structure to maximize profits in terms of performance and computational complexity. We go over several cases including LRMC of the PSD matrix~\cite{bart2009}, Euclidean distance matrix~\cite{luong2019}, and recommendation matrix~\cite{monti2017} and discuss how the special structure can be exploited in the algorithm design.

\subsubsection{Low-Rank PSD Matrix Completion}
%
In some applications, a desired matrix $\mathbf{M} \in \mathbb{R}^{n \times n}$ not only has a low-rank structure but also is positive semidefinite (i.e., $\mathbf{M} = \mathbf{M}^{T}$ and $\mathbf{z}^{T} \mathbf{M} \mathbf{z} \ge 0$ for any vector $\mathbf{z}$). In this case, the problem to recover $\mathbf{M}$ can be formulated as
\begin{equation} \label{eq:eq009}
\begin{split}
    \min\limits_\mathbf{X} &~~~~~ \text{rank}(\mathbf{X}) \\
    \text{subject to} &~~~~~ P_{\Omega} (\mathbf{X}) = P_{\Omega} (\mathbf{M}), \\
    				  &~~~~~ \mathbf{X} = \mathbf{X}^{T},~\mathbf{X} \succeq 0. 
\end{split}
\end{equation}
Similar to the rank minimization problem~\eqref{eq:rankmin}, the problem~\eqref{eq:eq009} can be relaxed using the nuclear norm, and the relaxed problem can be solved via SDP solvers.

The problem~\eqref{eq:eq009} can be simplified if the rank of a desired matrix is known in advance. Let $\text{rank}(\mathbf{M}) = k$. Then, since $\mathbf{M}$ is positive semidefinite, there exists a matrix $\mathbf{Z} \in \mathbb{R}^{n \times k}$ such that $\mathbf{M} = \mathbf{Z} \mathbf{Z}^{T}$. Using this, the problem~\eqref{eq:eq009} can be concisely expressed as 
\begin{equation} \label{eq:eq010}
\begin{split}
\min\limits_{\mathbf{Z}\in\mathbb{R}^{n\times k}} &~~~~~ \frac{1}{2}\| P_{\Omega} (\mathbf{M}) - P_{\Omega} (\mathbf{ZZ}^T) \|_{F}^2 . 
\end{split}
\end{equation}
Since~\eqref{eq:eq010} is an unconstrained optimization problem with a differentiable cost function, many gradient-based techniques such as steepest descent, conjugate gradient, and Newton methods can be applied. It has been shown that under suitable conditions of the coherence property of $\mathbf{M}$ and the number of the observed entries $|\Omega|$, the global convergence of gradient-based algorithms is guaranteed~\cite{ge2016}.

\vspace{.25cm}
\subsubsection{Euclidean Distance Matrix Completion}
%
Low-rank Euclidean distance matrix completion arises in the localization problem (e.g., sensor node localization in IoT networks). Let $\{\mathbf{z}_i\}_{i=1}^{n}$ be sensor locations in the $k$-dimensional Euclidean space ($k = 2$ or $k = 3$). Then, the Euclidean distance matrix $\mathbf{M} = (m_{ij}) \in \mathbb{R}^{n \times n}$ of sensor nodes is defined as $m_{ij} = \| \mathbf{z}_{i} - \mathbf{z}_{j} \|_{2}^{2}$. It is obvious that $\mathbf{M}$ is symmetric with diagonal elements being zero (i.e., $m_{ii}=0$). As mentioned, the rank of the Euclidean distance matrix $\mathbf{M}$ is at most $k+2$ (i.e., $\text{rank}(\mathbf{M}) \le k+2$). Also, one can show that a matrix $\mathbf{D} \in \mathbb{R}^{n \times n}$ is a Euclidean distance matrix if and only if $\mathbf{D} = \mathbf{D}^{T}$ and~\cite{dattorro2005} 
\begin{equation} \label{eq:eq011}
(\mathbf{I}_{n}-\frac{1}{n}\mathbf{hh}^T)\mathbf{D}(\mathbf{I}_{n}-\frac{1}{n}\mathbf{hh}^T)\preceq 0,
\end{equation}    
where $\mathbf{h} = [1 \ 1 \cdots \ 1]^{T} \in \mathbb{R}^{n}$. Using these, the problem to recover the Euclidean distance matrix $\mathbf{M}$ can be formulated as
\begin{equation} \label{eq:eq013}
\begin{split}
    \min\limits_\mathbf{D} &~~~~~ \|P_{\Omega} (\mathbf{D}) - P_{\Omega} (\mathbf{M})\|_F^2 \\
    \text{subject to}      &~~~~~ \text{rank}(\mathbf{D}) \leq k+2, \\
				    	   &~~~~~ \mathbf{D} = \mathbf{D}^T, \\
    					   &~~~~~ -(\mathbf{I}_{n}-\frac{1}{n}\mathbf{hh}^T)\mathbf{D}(\mathbf{I}_{n}-\frac{1}{n}\mathbf{hh}^T) \succeq 0.  
\end{split}
\end{equation}
Let $\mathbf{Y} = \mathbf{Z} \mathbf{Z}^{T}$ where $\mathbf{Z} = [ \mathbf{z}_{1} \ \cdots \ \mathbf{z}_{n} ]^{T} \in \mathbb{R}^{n \times k}$ is the matrix of sensor locations. Then, one can easily check that 
\begin{align} \label{eq:eq014}
\mathbf{M} = \text{diag}(\mathbf{Y})\mathbf{h}^T + \mathbf{h}\text{diag}(\mathbf{Y})^T - 2\mathbf{Y}.
\end{align}     
Thus, by letting $g(\mathbf{Y}) = \text{diag}(\mathbf{Y})\mathbf{h}^{T} + \mathbf{h}\text{diag}(\mathbf{Y})^{T} - 2\mathbf{Y}$, the problem in~\eqref{eq:eq013} can be equivalently formulated as
\begin{align} \label{eq:eq015}
\begin{split}
    \min\limits_\mathbf{Y} &~~~~~ \|P_{\Omega} (g(\mathbf{Y})) - P_{\Omega} (\mathbf{M})\|_F^2 \\
    \text{subject to} &~~~~~ \text{rank}(\mathbf{Y}) \leq k, \\
    				  &~~~~~ \mathbf{Y} = \mathbf{Y}^T,~\mathbf{Y} \succeq 0.
\end{split}  
\end{align}
Since the feasible set associated with the problem in~\eqref{eq:eq015} is a smooth Riemannian manifold~\cite{helmke1994, bart2009}, an extension of the Euclidean space on which a notion of differentiation exists~\cite{absil2009optimization, lee2003smooth}, various gradient-based optimization techniques such as steepest descent, Newton method, and conjugate gradient algorithms can be applied to solve~\eqref{eq:eq015}~\cite{luongITA,luong2019,absil2009optimization}.


\vspace{.25cm}
\subsubsection{Convolutional Neural Network Based Matrix Completion}


In recent years, approaches to use CNN to solve the LRMC problem have been proposed. These approaches are particular useful when a desired low-rank matrix is expressed as a graph model (e.g., the recommendation matrix with a user graph to express the similarity between users' rating results)~\cite{monti2017, sedhain2015autorec, zheng2016neural, he2017neural, defferrard2016, bruna2014spectral, henaff2015deep}.
The main idea of CNN-based LRMC algorithms is to express the low-rank matrix as a graph structure and then apply CNN to the constructed graph to recover the desired matrix.


\textit{Graphical Model of a Low-Rank Matrix}: Suppose $\mathbf{M} \in \mathbb{R}^{n_{1} \times n_{2}}$ is the rating matrix in which the columns and rows are indexed by users and products, respectively. The first step of the CNN-based LRMC algorithm is to model the column and row graphs of $\mathbf{M}$ using the correlations between its columns and rows. Specifically, in the column graph $\mathcal{G}_{c}$ of $\mathbf{M}$, users are represented as vertices, and two vertices $i$ and $j$ are connected by an undirected edge if the correlation $\rho_{ij} = \frac{|\langle \mathbf{m}_{i}, \mathbf{m}_{j} \rangle |}{\| \mathbf{m}_{i} \|_{2} \| \mathbf{m}_{j} \|_{2}}$ between  the $i$ and $j$-th columns of $\mathbf{M}$ is larger than the pre-determined threshold $\epsilon$. Similarly, we construct the row graph $\mathcal{G}_{r}$ of $\mathbf{M}$ by denoting each row (product) of $\mathbf{M}$ as a vertex and then connecting strongly correlated vertices. To express the connection, we define the adjacency matrix of each graph. The adjacency matrix $\mathbf{W}_{c} = (w^{c}_{ij}) \in \mathbb{R}^{n_{2} \times n_{2}}$ of the column graph $\mathcal{G}_{c}$ is defined as
\begin{align}
w^{c}_{ij}
&= \begin{cases}
1 & \text{if the vertices (users) $i$ and $j$ are connected} \\
0 & \text{otherwise}
\end{cases}
\end{align}
The adjacency matrix $\mathbf{W}_{r} \in \mathbb{R}^{n_{1} \times n_{2}}$ of the row graph $\mathcal{G}_{r}$ is defined in a similar way.


\textit{CNN-based LRMC}: Let $\mathbf{U} \in \mathbb{R}^{n_{1} \times r}$ and $\mathbf{V} \in \mathbb{R}^{n_{2} \times r}$ be matrices such that $\mathbf{M} = \mathbf{U} \mathbf{V}^{T}$. The primary task of the CNN-based approach is to find functions $f_{r}$ and $f_{c}$ mapping the vertex sets of the row and column graphs $\mathcal{G}_{r}$ and $\mathcal{G}_{c}$ of $\mathbf{M}$ to $\mathbf{U}$ and $\mathbf{V}$, respectively. Here, each vertex of $\mathcal{G}_r$ (respective $\mathcal{G}_c$) is mapped to each row of $\mathbf{U}$ (respective $\mathbf{V}$) by $f_r$ (respective $f_c$).
Since it is difficult to express $f_{r}$ and $f_{c}$ explicitly, we can learn these nonlinear mappings using CNN-based models. 
In the CNN-based LRMC approach, $\mathbf{U}$ and $\mathbf{V}$ are initialized at random and updated in each iteration. Specifically, $\mathbf{U}$ and $\mathbf{V}$ are updated to minimize the following loss function~\cite{monti2017}:
\begin{align}
l(\mathbf{U}, \mathbf{V})
&= \sum\limits_{(i,j):w_{ij}^{r} = 1}\| \mathbf{u}_{i}-\mathbf{u}_{j} \|_{2}^{2} + \sum\limits_{(i,j):w_{ij}^{c} = 1} \| \mathbf{v}_{i} - \mathbf{v}_{j} \|_{2}^{2} \nonumber \\
&~~+\frac{\tau}{2} \| P_{\Omega}(\sum\limits_{i=1}^r \mathbf{u}_{i} \mathbf{v}_{i}^{T}) - P_{\Omega}(\mathbf{M}) \|_{F}^{2},
\label{eq:eq021}
\end{align}
where $\tau$ is a regularization parameter. In other words, we find $\mathbf{U}$ and $\mathbf{V}$ such that the Euclidean distance between the connected vertices is minimized (see $\| \mathbf{u}_{i} - \mathbf{u}_{j} \|_{2}$ ($w_{ij}^{r}=1$) and $\| \mathbf{v}_{i} - \mathbf{v}_{j} \|_{2}$ ($w_{ij}^{c}=1$) in~\eqref{eq:eq021}). 
%
%
%
%
%
The update procedures of $\mathbf{U}$ and $\mathbf{V}$ are~\cite{monti2017}: 
\begin{enumerate}
\item[1)] Initialize $\mathbf{U}$ and $\mathbf{V}$ at random and assign each row of $\mathbf{U}$ and $\mathbf{V}$ to each vertex of the row graph $\mathcal{G}_{r}$ and the column graph $\mathcal{G}_{c}$, respectively.

\item[2)] Extract the feature matrices $\Delta \mathbf{U}$ and $\Delta \mathbf{V}$ by performing a graph-based convolution operation on $\mathcal{G}_{r}$ and $\mathcal{G}_{c}$, respectively.

\item[3)] Update $\mathbf{U}$ and $\mathbf{V}$ using the feature matrices $\Delta \mathbf{U}$ and $\Delta \mathbf{V}$, respectively. 

\item[4)] Compute the loss function in~\eqref{eq:eq021} using updated $\mathbf{U}$ and $\mathbf{V}$ and perform the back propagation to update the filter parameters.

\item[5)] Repeat the above procedures until the value of the loss function is smaller than a pre-chosen threshold.
\end{enumerate}


One important issue in the CNN-based LRMC approach is to define a graph-based convolution operation to extract the feature matrices $\Delta \mathbf{U}$ and $\Delta \mathbf{V}$ (see the second step). Note that the input data $\mathcal{G}_{r}$ and $\mathcal{G}_{c}$ do not lie on regular lattices like images and thus classical CNN cannot be directly applied to $\mathcal{G}_{r}$ and $\mathcal{G}_{c}$. 
One possible option is to define the convolution operation in the Fourier domain of the graph. 
In recent years, CNN models based on the Fourier transformation of graph-structure data have been proposed~\cite{bruna2014spectral,henaff2015deep,defferrard2016,kipf2016,montigeo2017}. In \cite{bruna2014spectral}, an approach to use the eigendecomposition of the Laplacian has been proposed. To further reduce the model complexity, CNN models using the polynomial filters have been proposed~\cite{defferrard2016,henaff2015deep,kipf2016}.
In essence, the Fourier transform of a graph can be computed using the (normalized) graph Laplacian. Let $\mathbf{R}_{r}$ be the graph Laplacian of $\mathcal{G}_{r}$ (i.e., $\mathbf{R}_{r} = \mathbf{I} - \mathbf{D}_{r}^{-1/2}\mathbf{W}_{r}\mathbf{D}_{r}^{-1/2}$ where $\mathbf{D}_{r} = \text{diag}(\mathbf{W}_{r} \mathbf{1}_{n_{2} \times 1})$)~\cite{hammond2011}. Then, the graph Fourier transform $\mathcal{F}_{r}(\mathbf{u})$ of a vertex assigned with the vector $\mathbf{u}$ is defined as 
\begin{align} \label{eq:def_graph Fourier transform}
\mathcal{F}_{r}(\mathbf{u}) 
&= \mathbf{Q}_{r}^{T} \mathbf{u},
\end{align}
where $\mathbf{R}_{r} = \mathbf{Q}_{r} \mathbf{\Lambda}_{r} \mathbf{Q}_{r}^{T}$ is an eigen-decomposition of the graph Laplacian $\mathbf{R}_{r}$~\cite{hammond2011}. Also, the inverse graph Fourier transform $\mathcal{F}_{r}^{-1}(\mathbf{u}^{\prime})$ of $\mathbf{u}'$ is defined as\footnote{One can easily check that $\mathcal{F}_{r}^{-1}(\mathcal{F}_{r}(\mathbf{u})) = \mathbf{u}$ and $\mathcal{F}_{r}(\mathcal{F}_{r}^{-1}(\mathbf{u}')) = \mathbf{u}'$.}
\begin{align} \label{eq:def_graph inverse Fourier transform}
\mathcal{F}_{r}^{-1}(\mathbf{u}')
&= \mathbf{Q}_{r} \mathbf{u}'.
\end{align}
Let $\mathbf{z}$ be the filter used in the convolution, then the output $\Delta \mathbf{u}$ of the graph-based convolution on a vertex assigned with the vector $\mathbf{u}$ is defined as \cite{hammond2011, defferrard2016}
\begin{align} \label{eq:eq101}
\Delta \mathbf{u} &= \mathbf{z}\ast\mathbf{u} = \mathcal{F}_{r}^{-1}(\mathcal{F}_{r}(\mathbf{z}) \odot \mathcal{F}_{r}(\mathbf{u}))
\end{align}
From~\eqref{eq:def_graph Fourier transform} and~\eqref{eq:def_graph inverse Fourier transform},~\eqref{eq:eq101} can be expressed as
\begin{align}
\Delta \mathbf{u} 
&= \mathbf{Q}_r(\mathcal{F}_r(\mathbf{z})\odot \mathbf{Q}_r^T\mathbf{u}) \nonumber \\
&= \mathbf{Q}_r \text{diag}(\mathcal{F}_r(\mathbf{z}))\mathbf{Q}_r^T\mathbf{u} \nonumber \\
&= \mathbf{Q}_r \mathbf{G}\mathbf{Q}_r^T\mathbf{u},
\label{eq:eq105}
\end{align}
where $\mathbf{G} = \text{diag}(\mathcal{F}_r(\mathbf{z}))$ is the matrix of filter parameters defined in the graph Fourier domain.

We next update $\mathbf{U}$ and $\mathbf{V}$ using the feature matrices $\Delta \mathbf{U}$ and $\Delta \mathbf{V}$. In \cite{monti2017}, a cascade of multi-graph CNN followed by long short-term memory
(LSTM) recurrent neural network has been proposed. 
The computational cost of this approach is $\mathcal{O}(r|\Omega| + r^2n_1 + r^2n_2)$ which is much lower than the SVD-based LRMC techniques (i.e., $\mathcal{O}(rn_1n_2)$) as long as $r\ll\min(n_1,n_2)$.
Finally, we compute the loss function $l (\mathbf{U}_{i}, \mathbf{V}_{i})$ in \eqref{eq:eq021} and then update the filter parameters using the backpropagation.
Suppose $\{\mathbf{U}_{i}\}_{i}$ and $\{\mathbf{V}_{i}\}_{i}$ converge to $\widehat{\mathbf{U}}$ and $\widehat{\mathbf{V}}$, respectively, then the estimate of $\mathbf{M}$ obtained by the CNN-based LRMC is $\widehat{\mathbf{M}} = \widehat{\mathbf{U}}\widehat{\mathbf{V}}^T$.

 
\vspace{.25cm}
\subsubsection{Atomic Norm Minimization}
%

In ADMiRA, a low-rank matrix can be represented using a small number of rank-one matrices. Atomic norm minimization (ANM) generalizes this idea for arbitrary data in which the data is represented using a small number of basis elements called \textit{atom}.
Example of ANM include sound navigation ranging systems~\cite{chandrasekaran2012} and line spectral estimation~\cite{bhaskar2013}.
To be specific, let $\mathbf{X}=\sum\limits_{i = 1}^r \alpha_i\mathbf{H}_i$ be a signal with $k$ distinct frequency components $\mathbf{H}_i\in\mathbb{C}^{n_1\times n_2}$. Then the atom is defined as $\mathbf{H}_i = \mathbf{h}_i\mathbf{b}^\ast$ where 
\begin{align}
\mathbf{h}_{i} 
&= [1 \ e^{j2\pi f_i} \ e^{j2\pi f_{i} 2} \ \cdots \ e^{j2 \pi f_{i} (n_{1}-1) }]^{T}
\end{align}
is the steering vector and $\mathbf{b}_i\in\mathbb{C}^{n_2}$ is the vector of normalized coefficients (i.e., $\|\mathbf{b}_i\|_2 = 1$).
We denote the set of such atoms $\mathbf{H}_i$ as $\mathcal{H}$. Using $\mathcal{H}$, the atomic norm of $\mathbf{X}$ is defined as
\begin{equation}
\|\mathbf{X}\|_{\mathcal{H}} = \inf \{\sum\limits_{i}\alpha_i : \mathbf{X} = \sum\limits_{i}\alpha_i \mathbf{H}_i,\: \alpha_i> 0,~\mathbf{H}_i\in\mathcal{H} \} .
\label{eq:eq0122}
\end{equation} 
Note that the atomic norm $\|\mathbf{X}\|_{\mathcal{H}}$ is a generalization of the $\ell_1$-norm and also the nuclear norm to the space of sinusoidal signals~\cite{bhaskar2013,choi2017}. 


\begin{table}[t]
\centering
\caption{Summary of the LRMC algorithms. The rank is $r$ and $n = \max(n_1, n_2)$.}
\label{tab:tab001}
\begin{threeparttable}
\begin{tabular}{|p{1.1cm}|p{1.6cm}|p{1.7cm}|p{4.5cm}|p{2cm}|p{2cm}|}
\hline
Category & Technique & Algorithm & Features &  \parbox{2cm}{\Tstrut Computational \\ Complexity  \Bstrut} &  \parbox{2cm}{\Tstrut Iteration \\ Complexity\tnote{*}  \Bstrut} \\
\hline
\multirow{4}{1cm}{\centering NNM} & { Convex \qquad Optimization } & { SDPT3\qquad (CVX) \cite{toh} }  &  A solver for conic programming problems  & $\mathcal{O}(n^3)$ & $\mathcal{O}(n^\omega\log(\frac{1}{\epsilon}))$\tnote{**}\\
\cline{2-6}
& \multirow{2}{1.6cm}{\Tstrut NNM via Singular Value Thresholding} & SVT \cite{svt}  &  An extension of the iterative soft thresholding technique in compressed sensing for LRMC, based on a Lagrange multiplier method & $\mathcal{O}(rn^2)$ & $\mathcal{O}(\frac{1}{\sqrt{\epsilon}})$\\
\cline{3-6}
&  & NIHT \cite{tanner}  & An extension of the iterative hard thresholding technique \cite{blumensath} in compressed sensing for LRMC & $\mathcal{O}(rn^2)$ & $\mathcal{O}(\log(\frac{1}{\epsilon}))$\\
\cline{2-6}
 & IRLS \qquad Minimization  &  IRLS-M \quad\qquad Algorithm~\cite{fornasier2011}  & An algorithm to solve the NNM problem by computing the solution of a weighted least squares subproblem in each iteration   & $\mathcal{O}(rn^2)$ & $\mathcal{O}(\log(\frac{1}{\epsilon}))$ \\
\cline{1-6}
\multirow{7}{1cm}{\centering FNM with Rank Constraint} &  Greedy\quad  Technique  & ADMiRA \cite{adm}  &  An extension of the greedy algorithm CoSaMP \cite{needell2009,choi2017} in compressed sensing for LRMC, uses greedy projection to identify a set of rank-one matrices that best represents the original matrix  & $\mathcal{O}(rn^2)$ & $\mathcal{O}(\log(\frac{1}{\epsilon}))$\\
\cline{2-6}
& \multirow{2}{1.5cm}{\Tstrut Alternating Minimization} & LMaFit \cite{lmafit} &  A nonlinear successive over-relaxation LRMC algorithm based on nonlinear Gauss-Seidel method  & $\mathcal{O}(r|\Omega|+r^2n)$& $\mathcal{O}(\log(\frac{1}{\epsilon}))$\\
\cline{3-6}
& & ASD \cite{asd}  &  A steepest decent algorithm for the FNM-based LRMC problem (\ref{eq:lsm})  & $\mathcal{O}(r|\Omega|+r^2n)$ & $\mathcal{O}(\log(\frac{1}{\epsilon}))$\\
\cline{2-6}
& \multirow{2}{1.5cm}{\Tstrut Manifold Optimization} 
& SET \cite{dai2010}  &  A gradient-based algorithm to solve the FNM problem on a Grassmann manifold  & $\mathcal{O}(r|\Omega|+r^2n)$ & $\mathcal{O}(\log(\frac{1}{\epsilon}))$\\
\cline{3-6}
& & LRGeomCG \cite{bart2013}  &  A conjugate gradient algorithm over a Riemannian manifold of the fixed-rank matrices & $\mathcal{O}(r|\Omega|+r^2n)$ & $\mathcal{O}(\log(\frac{1}{\epsilon}))$\\
\cline{2-6}
 & \multirow{2}{1.5cm}{Truncated NNM}  &  TNNR-\quad APGL \cite{truncatedNNM}  &  This algorithm solves the truncated NNM problem via accelerated proximal gradient line search method \cite{beck2009}  & $\mathcal{O}(rn^2)$ & $\mathcal{O}(\frac{1}{\sqrt{\epsilon}})$\\
\cline{3-6}
& & TNNR-ADMM \cite{truncatedNNM}  & This algorithm solves the truncated NNM problem via an alternating direction method of multipliers \cite{tao2011}  & $\mathcal{O}(rn^2)$ & $\mathcal{O}(\frac{1}{\sqrt{\epsilon}})$\\
\cline{2-6}
& CNN-based Technique  &  CNN-based LRMC Algorithm~\cite{monti2017}  & An gradient-based algorithm to express a low-rank matrix as a graph structure and then apply CNN to the constructed graph to recover the desired matrix   & $\mathcal{O}(r|\Omega|+r^2n)$ & $\mathcal{O}(\log(\frac{1}{\epsilon}))$\\
\hline
\end{tabular}
\begin{tablenotes}
\item[*] The number of iterations to achieve the reconstructed matrix $\widehat{\mathbf{M}}$ satisfying $\|\widehat{\mathbf{M}} - \mathbf{M}^\ast\|_F\leq \epsilon$ where $\mathbf{M}^\ast$ is the optimal solution.
\item[**] $\omega$ is some positive constant controlling the iteration complexity. 
\end{tablenotes}
\end{threeparttable}
\end{table}

Let $\mathbf{X}_o$ be the observation of $\mathbf{X}$, then the problem to reconstruct $\mathbf{X}$ can be modeled as the ANM problem: 
\begin{equation}
\begin{matrix}
\min\limits_{\mathbf{Z}} & \frac{1}{2}\|\mathbf{Z}-\mathbf{X}_o\|_F + \tau\|\mathbf{Z}\|_{\mathcal{H}}, 
\end{matrix} 
\label{eq:eq008}
\end{equation} 
where $\tau > 0$ is a regularization parameter. By using~\cite[Theorem 1]{li2016off}, we have
\begin{align}
\| \mathbf{Z} \|_{\mathcal{H}}
&= \inf \left \{ \frac{1}{2} \left ( \text{tr}(\mathbf{W}) + \text{tr}(\text{Toep}(\mathbf{u})) \right ): \right . \notag\\ 
& \qquad\qquad \left . \begin{bmatrix}
\mathbf{W} & \mathbf{Z}^{*} \\
\mathbf{Z} & \text{Toep}(\mathbf{u})
\end{bmatrix} \succeq 0 \right \},
\end{align}
and the equivalent expression of the problem (\ref{eq:eq008}) is
\begin{equation} \label{eq:atomic norm minimization problem_SDP formulation}
\begin{split}
&\underset{\mathbf{Z}, \mathbf{W}, \mathbf{u}}{\min} ~~~ \Vert \mathbf{Z}-\mathbf{X}_o \Vert^2_{F} + \tau \left ( \text{tr}(\mathbf{W}) + \text{tr}(\text{Toep}(\mathbf{u})) \right ) \\
&~~\text{s.t.}~~~~ \begin{bmatrix}
\mathbf{W} & \mathbf{Z}^{*} \\
\mathbf{Z} & \text{Toep}(\mathbf{u})
\end{bmatrix} \succeq 0
\end{split}.
\end{equation} 
Note that the problem~\eqref{eq:atomic norm minimization problem_SDP formulation} can be solved via the SDP solver (e.g., SDPT3~\cite{toh}) or greedy algorithms \cite{chen2001,rao2015}.

\vspace{.25cm}
\section{Numerical Evaluation}
%

In this section, we study the performance of the LRMC algorithms. In our experiments, we focus on the algorithm listed in Table \ref{tab:tab001}. The original matrix is generated by the product of two random matrices $\mathbf{A}\in\mathbb{R}^{n_1\times r}$ and $\mathbf{B}\in\mathbb{R}^{n_2\times r}$, i.e., $\mathbf{M} = \mathbf{AB}^T$. Entries of these two matrices, $a_{ij}$ and $b_{pq}$ are identically and independently distributed random variables sampled from the normal distribution $\mathcal{N}(0,1)$. Sampled elements are also chosen at random. The sampling ratio $p$ is defined as
\begin{equation}
p = \frac{|\Omega|}{n_1n_2}\nonumber,
\end{equation}
where $|\Omega|$ is the cardinality (number of elements) of $\Omega$. In the noisy scenario, we use the additive noise model where the observed matrix $\mathbf{M}_o$ is expressed as $\mathbf{M}_o = \mathbf{M} + \mathbf{N}$ where the noise matrix $\mathbf{N}$ is formed by the i.i.d. random entries sampled from the Gaussian distribution $\mathcal{N}(0,\sigma^2)$. For given SNR, $\sigma^2 = \frac{1}{n_1n_2}\|\mathbf{M}\|_F^210^{-\frac{\text{SNR}}{10}}$. Note that the parameters of the LRMC algorithm are chosen from the reference paper. For each point of the algorithm, we run $1,000$ independent trials and then plot the average value.

In the performance evaluation of the LRMC algorithms, we use the mean square error (MSE) and the exact recovery ratio, which are defined, respectively, as
\begin{equation}
\begin{split}
& \text{MSE} = \frac{1}{n_{1}n_{2}}\|\widehat{\mathbf{M}}- \mathbf{M}\|^2_F, \\
& R = \frac{\text{number of successful trials}}{\text{total trials}},
\end{split} \nonumber
\end{equation}
where $\widehat{\mathbf{M}}$ is the reconstructed low-rank matrix. We say the trial is successful if the MSE performance is less than the threshold $\epsilon$. In our experiments, we set $\epsilon = 10^{-6}$. Here, $R$ can be used to represent the probability of successful recovery.
 
We first examine the exact recovery ratio of the LRMC algorithms in terms of the sampling ratio and the rank of $\mathbf{M}$. In our experiments, we set $n_1 = n_2 = 100$ and compute the phase transition \cite{donoho2009} of the LRMC algorithms. 
Note that the phase transition is a contour plot of the success probability $P$ (we set $P = 0.5$) where the sampling ratio ($x$-axis) and the rank ($y$-axis) form a regular grid of the $x$-$y$ plane. The contour plot separates the plane into two areas: the area above the curve is one satisfying $P < 0.5$ and the area below the curve is a region achieving $P > 0.5$~\cite{donoho2009} (see Fig. \ref{fig:fig002}). 
The higher the curve, therefore, the better the algorithm would be. In general, the LRMC algorithms perform poor when the matrix has a small number of observed entries and the rank is large. Overall, NNM-based algorithms perform better than FNM-based algorithms. In particular, the NNM technique using SDPT3 solver outperforms the rest because the convex optimization technique always finds a global optimum while other techniques often converge to local optimum.          


\begin{figure*} [t!]
\label{fig:fig002}
\centering    
\includegraphics[width=80mm,height=75mm]{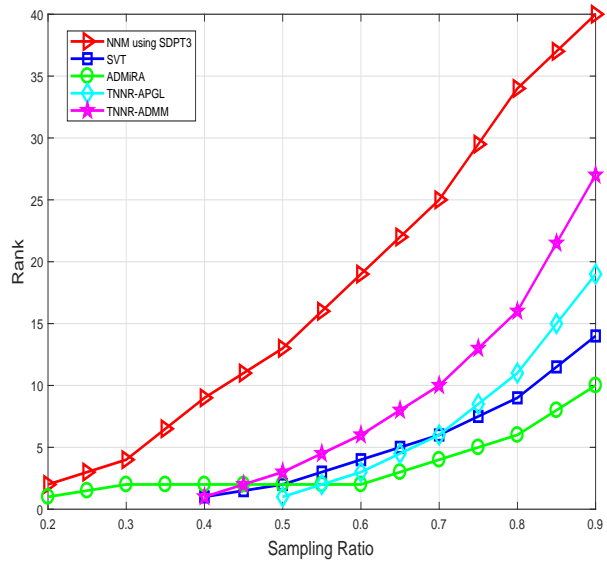}  
\caption {Phase transition of LRMC algorithms.} 
\end{figure*}

In order to investigate the computational efficiency of LRMC algorithms, we measure the running time of each algorithm as a function of rank (see Fig. \ref{fig:fig003}). The running time is measured in second, using a 64-bit PC with an Intel i5-4670 CPU running at 3.4 GHz. We observe that the convex algorithms have a relatively high running time complexity.



\begin{figure*} [t!]
\label{fig:fig003}
\centering    
\includegraphics[width=80mm,height=75mm]{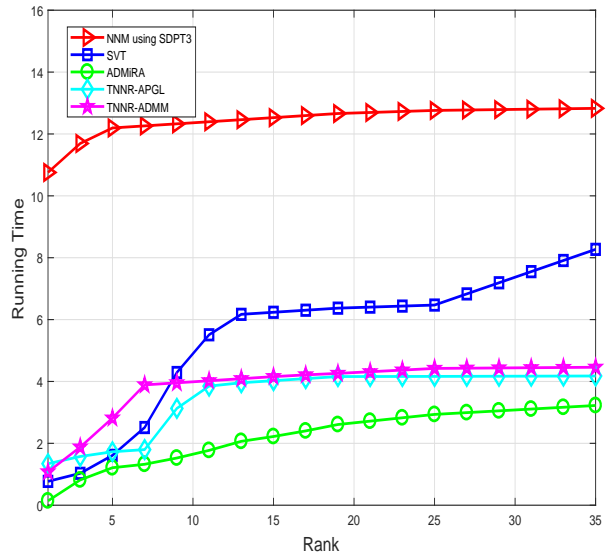}  
\caption {Running times of LRMC algorithms in noiseless scenario (40\% of entries are observed).} 
\end{figure*}


We next examine the efficiency of the LRMC algorithms for different problem size (see Table \ref{tab:tab002}). For iterative LRMC algorithms, we set the maximum number of iteration to 300. We see that LRMC algorithms such as SVT, IRLS-M, ASD, ADMiRA, and LRGeomCG run fast. For example, it takes less than a minute for these algorithms to reconstruct $1000\times 1000$ matrix, while the running time of SDPT3 solver is more than 5 minutes. Further reduction of the running time can be achieved using the alternating projection-based algorithms such as LMaFit. For example, it takes about one second to reconstruct an $(1000\times 1000)$-dimensional matrix with rank $5$ using LMaFit. Therefore, when the exact recovery of the original matrix is unnecessary, the FNM-based technique would be a good choice.  

\begin{table}[t]
\centering
\caption{MSE results for different problem sizes where $\text{rank}(\mathbf{M}) = 5$, and $p = 2\times \text{DOF}$}
\label{tab:tab002}
\begin{tabular}{|p{2.4cm}|p{1cm}|p{1cm}|p{1.3cm}|p{1cm}|p{1cm}|p{1.3cm}|p{1cm}|p{1cm}|p{1.3cm}|}
\hline
\multirow{2}{2cm}{} & \multicolumn{3}{c|}{\centering $n_1 = n_2 = 50$} & \multicolumn{3}{c|}{\centering $n_1 = n_2 = 500$ } & \multicolumn{3}{c|}{\centering $n_1 = n_2 = 1000$} \\
\cline{2-10}
& MSE & Running Time (s) & Number of Iterations & MSE  & Running Time (s) & Number of Iterations & MSE  &  Running Time (s) & Number of Iterations \\
\hline
NNM using SDPT3  	& 0.0072 & 0.6 	 & 13  & 0.0017 & 74  & 16  & 0.0010 & 354 & 16\\
\hline
SVT  								& 0.0154 & 0.4 	 & 300 & 0.4564 & 10  & 300 & 0.2110 & 32  & 300\\
\hline
NIHT 								& 0.0008 & 0.2  	 & 253 & 0.0039 & 21  & 300 & 0.0019 & 93  & 300\\
\hline
IRLS-M 								& 0.0009 & 0.2 	 & 60  & 0.0033 & 2   & 60  & 0.0025 & 8   & 60 \\
\hline
ADMiRA 								& 0.0075 & 0.3 	 & 300 & 0.0029 & 49  & 300 & 0.0016 & 52 & 300\\
\hline
ASD 								& 0.0003 & $10^{-2}$ 	 & 227 & 0.0006 & 2   & 300 & 0.0005 & 8   & 300\\
\hline
LMaFit 								& 0.0002 & $10^{-2}$ & 241 & 0.0002 & 0.5 & 300 & 0.0500 & 1   & 300\\
\hline
SET 								& 0.0678 & 11 	 & 9   & 0.0260 & 136 & 8   & 0.0108 & 270 & 8\\
\hline
LRGeomCG 							& 0.0287 & 0.1 	 & 108  & 0.0333 & 12  & 300 & 0.0165 & 40  & 300\\
\hline
TNNR-ADMM 							& 0.0221 & 0.3 	 & 300 & 0.0042 & 22  & 300 & 0.0021 & 94  & 300\\
\hline
TNNR-APGL  							& 0.0055 & 0.3 	 & 300 & 0.0011 & 21  & 300 & 0.0009 & 95  & 300\\
\hline
\end{tabular}
\end{table}

In the noisy scenario, we also observe that FNM-based algorithms perform well (see Fig. \ref{fig:fig005} and Fig. \ref{fig:fig0055}). In this experiment, we compute the MSE of LRMC algorithms against the rank of the original low-rank matrix for different setting of SNR (i.e., SNR = 20 and 50 dB).  
We observe that in the low and mid SNR regime (e.g., SNR = 20 dB), TNNR-ADMM performs comparable to the NNM-based algorithms since the FNM-based cost function is robust to the noise. In the high SNR regime (e.g., SNR = 50 dB), the convex algorithm (NNM with SDPT3) exhibits the best performance in term of the MSE. The performance of TNNR-ADMM is notably better than that of the rest of LRMC algorithms. For example, given $\text{rank}(\mathbf{M}) = 20$, the MSE of TNNR-ADMM is around 0.04, while the MSE of the rest is higher than $1$.



\begin{figure*} [t!]
\label{fig:fig005}
\centering    
\includegraphics[width=80mm,height=75mm]{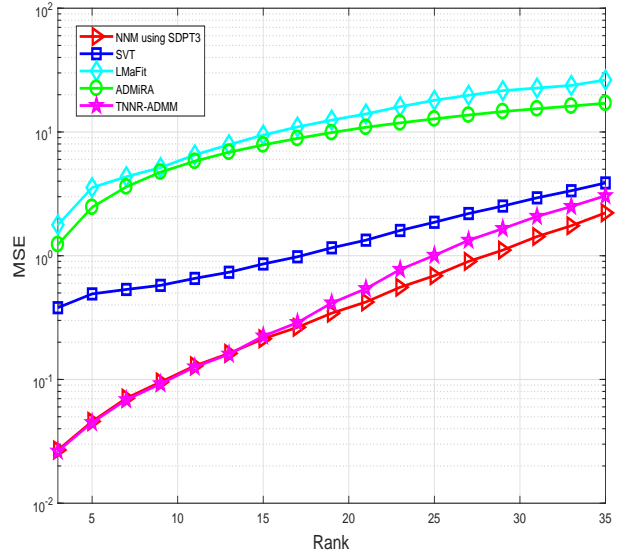}   
\caption {MSE performance of LRMC algorithms in noisy scenario with SNR = 20 dB (70\% of entries are observed).} 
\end{figure*}


\begin{figure*} [t!]
\label{fig:fig0055}
\centering    
\includegraphics[width=80mm,height=75mm]{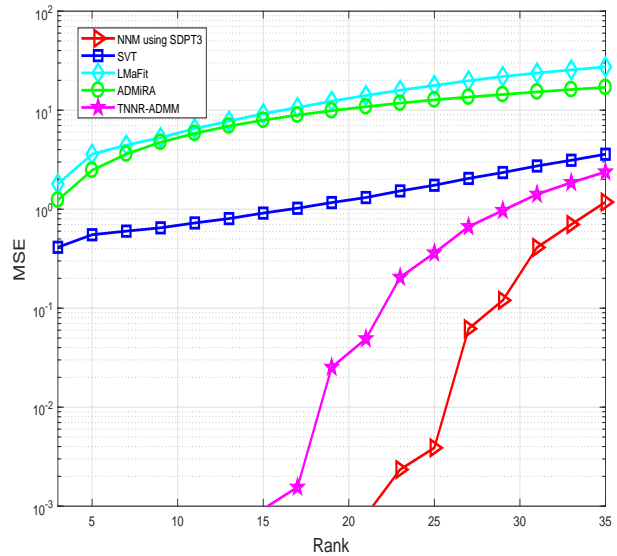}    
\caption {MSE performance of LRMC algorithms in noisy scenario with SNR = 50 dB (70\% of entries are observed).} 
\end{figure*}

Finally, we apply LRMC techniques to recover images corrupted by impulse noise. In this experiment, we use $256 \times 256$ standard grayscale images (e.g., boat, cameraman, lena, and pepper images) and the salt-and-pepper noise model with different noise density $\rho = 0.3, 0.5,\text{ and } 0.7$. For the FNM-based LRMC techniques, the rank is given by the number of the singular values $\sigma_i$ being greater than a relative threshold $\epsilon > 0$, i.e., $\sigma_i > \epsilon\max\limits_i\sigma_i$. From the simulation results, we observe that peak SNR (pSNR), defined as the ratio of the maximum pixel value of the image to noise variance, of all LRMC techniques is at least $52$dB when $\rho = 0.3$ (see Table \ref{tab:tab004}). In particular, NNM using SDPT3, SVT, and IRLS-M outperform the rest, achieving pSNR$ \geq 57$ dB even with high noise level $\rho = 0.7$.

\begin{table}[t]
\centering
\caption{Image recovery via LRMC for different noise levels $\rho$.}
\label{tab:tab004}
\begin{tabular}{|p{1.8cm}|p{1cm}|p{1cm}|p{1cm}|p{1cm}|p{1cm}|p{1cm}|p{1cm}|p{1cm}|p{1cm}|}
\hline
\multirow{2}{2cm}{} & \multicolumn{3}{c|}{\centering $\rho = 0.3$} & \multicolumn{3}{c|}{\centering $\rho = 0.5$} & \multicolumn{3}{c|}{\centering $\rho = 0.7$ }  \\
\cline{2-10}
& pSNR (dB)  & Running Time (s) & Iteration & pSNR (dB)  & Running Time (s) & Iteration & pSNR (dB)  & Running Time (s) & Iteration\\
\hline
NNM \qquad\quad using SDPT3	& 66   	& 1801 & 14 	& 59   	& 883  & 14	 & 58  & 292  & 15\\
\hline
SVT 						& 61   	& 18   & 300 	& 59   	& 13   & 300 & 57  & 32  & 300\\
\hline
NIHT 						& 58    & 16   & 300 	& 54   	& 6    & 154 & 53  & 2 & 35\\
\hline
IRLS-M 						& 68    & 87   & 60  	& 63   	& 43   & 60  & 59  & 15   & 60\\
\hline
ADMiRA 						& 57   	& 1391 & 300 	& 54   	& 423  & 245 & 53  & 210  & 66\\
\hline
ASD 						& 57   	& 3    & 300 	& 55   	& 4    & 300 & 53  & 2 & 101\\
\hline
LMaFit 						& 58   	& 2    & 300 	& 55   	& 1    & 123 & 53  & 0.3 & 34\\
\hline
SET 						& 61   	& 716  & 6 		& 55   	& 321  & 4 	 & 53  & 5   & 2\\
\hline
LRGeomCG 					& 52   	& 47   & 300 	& 48    & 18   & 75  & 44  & 5 & 21\\
\hline
TNNR-ADMM 					& 57   	& 15   & 300 	& 54   	& 18   & 300 & 53  & 18   & 300\\
\hline
TNNR-APGL 					& 56   	& 14   & 300 	& 56   	& 19   & 300 & 53  & 17   & 300\\
\hline
\end{tabular}
\end{table}

\section{Concluding Remarks}
\label{sec:conclusion}
In this paper, we presented a contemporary survey of LRMC techniques. Firstly, we classified state-of-the-art LRMC techniques into two main categories based on the availability of the rank information. Specifically, when the rank of a desired matrix is unknown, we formulated the LRMC problem as the NNM problem and discussed several NNM-based LRMC techniques such as SDP-based NNM, SVT, and truncated NNM. When the rank of an original matrix is known a priori, the LRMC problem can be modeled as the FNM problem. We discussed various FNM-based LRMC techniques (e.g., greedy algorithms, alternating projection methods, and optimization over Riemannian manifold). Secondly, we discussed fundamental issues and principles that one needs to be aware of when solving the LRMC problem. Specifically, we discussed two key properties, sparsity of observed entries and the incoherence of an original matrix, characterizing the LRMC problem. We also explained how to exploit the special structure of the desired matrix (e.g., PSD, Euclidean distance, and graph structures) in the LRMC algorithm design. Finally, we compared the performance of LRMC techniques via numerical simulations and provided the running time and computational complexity of each technique.

When one tries to use the LRMC techniques, a natural question one might ask is what algorithm should one choose? While this question is in general difficult to answer, one can consider the SDP-based NNM technique when the accuracy of a recovered matrix is critical (see Fig.~\ref{fig:fig002},~\ref{fig:fig005}, and~\ref{fig:fig0055}). Another important point that one should consider in the choice of the LRMC algorithm is the computational complexity. In many practical applications, dimension of a matrix to be recovered is large, and in fact in the order of hundred or thousand. In such large-scale problems, algorithms such as LMaFit and LRGeomCG would be a good option since their computational complexity scales linearly with the number of observed entries $\mathcal{O}(r|\Omega|)$ while the complexity of SDP-based NNM is expressed as $\mathcal{O}(n^{3})$ (see Table~\ref{tab:tab001}). In general, there is a trade-off between the running time and the recovery performance. In fact, FNM-based LRMC algorithms such as LMaFit and ADMiRA converge much faster than the convex optimization based algorithms (see Table~\ref{tab:tab002}), but the NNM-based LRMC algorithms are more reliable than FNM-based LRMC algorithms (see Fig.~\ref{fig:fig002} and~\ref{fig:fig0055}). So, one should consider the given setup and operating condition to obtain the best trade-off between the complexity and the performance.

We conclude the paper by providing some of future research directions.
\begin{itemize}
\item When the dimension of a low-rank matrix increases and thus computational complexity increases significantly, we want an algorithm with good recovery guarantee yet its complexity scales linearly with the problem size. Without doubt, in the real-time applications such as IoT localization and massive MIMO, low-complexity and short running time are of great importance. Development of implementation-friendly algorithm and architecture would accelerate the dissemination of LRMC techniques.


\item Most of the LRMC techniques assume that the original low-rank matrix is a random matrix whose entries are randomly generated. In many practical situations, however, entries of the matrix are not purely random but chosen from a finite set of integer numbers. In the recommendation system, for example, each entry (rating for a product) is chosen from integer value (e.g., $1\sim 5$ scale). Unfortunately, there is no well-known practical guideline and efficient algorithm when entries of a matrix are chosen from the discrete set. It would be useful to come up with a simple and effective LRMC technique suited for such applications.  

\item As mentioned, CNN-based LRMC technique is a useful tool to reconstruct a low-rank matrix. In essence, unknown entries of a low-rank matrix are recovered based on the graph model of the matrix. Since observed entries can be considered as labeled training data, this approach can be classified as a supervised learning. In many practical scenarios, however, it might not be easy to precisely express the graph model of the matrix since there are various criteria to define the graph edge. In addressing this problem, new deep learning technique such as the generative adversarial networks (GAN) \cite{goodfellow2016} consisting of the generator and discriminator would be useful.    
\end{itemize}

\appendices

\section{Proof of the SDP form of NNM} \label{app:appA}
\proof
We recall that the standard form of an SDP is expressed as
\begin{equation}
\begin{split}
\min\limits_{\mathbf{Y}} &~~~~~ \langle \mathbf{C},\mathbf{Y} \rangle \\
\text{subject to} &~~~~~ \langle \mathbf{A}_k,\mathbf{Y} \rangle \leq b_k,~~  k=1, \cdots, l\\
&~~~~~ \mathbf{Y}\succeq 0
\end{split}
\label{eq:eqSDP_001}
\end{equation}
where $\mathbf{C}$ is a given matrix, and $\{ \mathbf{A}_{k} \}_{k=1}^{l}$ and $\{ b_{k} \}_{k=1}^{l}$ are given sequences of matrices and constants, respectively.
To convert the NNM problem in (\ref{eq:nnm}) into the standard SDP form in~\eqref{eq:eqSDP_001}, we need a few steps. First, we convert the NNM problem in~\eqref{eq:nnm} into the epigraph form:\footnote{Note that $\underset{\mathbf{X}}{\min} \| \mathbf{X} \|_{*} = \underset{\mathbf{X}}{\min} \left \lbrace \underset{t: \| \mathbf{X} \|_{*} \le t}{\min} t \right \rbrace = \underset{(\mathbf{X}, t): \| \mathbf{X} \|_{*} \le t}{\min} t$.}
\begin{equation}\label{eq:eqSDP_002}
\begin{split}
\min\limits_{\mathbf{X}, t} &~~~~~t \\
\text{subject to}  &~~~~~\| \mathbf{X} \|_{*} \le t, \\
&~~~~~P_{\Omega}(\mathbf{X}) = P_{\Omega}(\mathbf{M}).
\end{split}
\end{equation}

Next, we transform the constraints in~\eqref{eq:eqSDP_002} to generate the standard form in~\eqref{eq:eqSDP_001}. We first consider the inequality constraint ($\| \mathbf{X} \|_{*} \le t$). Note that $\| \mathbf{X} \|_{*} \le t$ if and only if there are symmetric matrices $\mathbf{W}_{1} \in \mathbb{R}^{n_{1} \times n_{1}}$ and $\mathbf{W}_{2} \in \mathbb{R}^{n_{2} \times n_{2}}$ such that~\cite[Lemma 2]{fazel} 
\begin{align}
\text{tr} (\mathbf{W}_{1}) + \text{tr} (\mathbf{W}_{2}) \le 2t
~~\text{and}~~
\begin{bmatrix}
\mathbf{W}_{1} & \mathbf{X} \\
\mathbf{X}^{T} & \mathbf{W}_{2}
\end{bmatrix} \succeq 0.
\end{align}
Then, by denoting $\mathbf{Y} = \begin{bmatrix}
\mathbf{W}_{1} & \mathbf{X} \\
\mathbf{X}^{T} & \mathbf{W}_{2}
\end{bmatrix} \in \mathbb{R}^{(n_{1} + n_{2}) \times (n_{1} + n_{2})}$ and $\widetilde{\mathbf{M}} = \begin{bmatrix}
\mathbf{0}_{n_{1} \times n_{1}} & \mathbf{M} \\
\mathbf{M}^{T} & \mathbf{0}_{n_{2} \times n_{2}}
\end{bmatrix}$ where $\mathbf{0}_{s \times t}$ is the $(s \times t)$-dimensional zero matrix, the problem in~\eqref{eq:eqSDP_002} can be reformulated as
\begin{equation}\label{eq:eqSDP_003}
\begin{split}
\underset{\mathbf{Y}, t}{\min} &~~~~~~ 2t \\
\text{subject to}  &~~~~~~\text{tr} (\mathbf{Y}) \le 2t, \\
&~~~~~~ \mathbf{Y} \succeq 0, \\
&~~~~~~ P_{\widetilde{\Omega}}(\mathbf{Y}) = P_{\widetilde{\Omega}} ( \widetilde{\mathbf{M}} ),
\end{split}
\end{equation}
where $P_{\widetilde{\Omega}}(\mathbf{Y}) = \begin{bmatrix}
\mathbf{0}_{n_{1} \times n_{1}} & P_{\Omega} (\mathbf{X}) \\
(P_{\Omega}(\mathbf{X}))^{T} & \mathbf{0}_{n_{2} \times n_{2}}
\end{bmatrix}$ is the extended sampling operator. We now consider the equality constraint ($P_{\widetilde{\Omega}}(\mathbf{Y}) = P_{\widetilde{\Omega}}(\widetilde{\mathbf{M}})$) in~\eqref{eq:eqSDP_003}. One can easily see that this condition is equivalent to 
\begin{align}
\langle \mathbf{Y}, \mathbf{e}_{i} \mathbf{e}_{j + n_{1}}^{T} \rangle 
&= \langle \widetilde{\mathbf{M}}, \mathbf{e}_{i} \mathbf{e}_{j+n_{1}}^{T} \rangle,
~~~ (i, j) \in \Omega, \label{eq:equality constraint_SDP_1}
\end{align}
where $\{ \mathbf{e}_{1}, \cdots, \mathbf{e}_{n_{1} + n_{2}} \}$ be the standard ordered basis of $\mathbb{R}^{n_{1} + n_{2}}$. Let $\mathbf{A}_k = \mathbf{e}_i \mathbf{e}_{j+n_{1}}^T$ and $\langle \widetilde{\mathbf{M}}, \mathbf{e}_i \mathbf{e}_{j+n_{1}}^T \rangle = b_k$ for each of $(i,j)\in\Omega$. Then,
\begin{align} \label{eq:equality constraint_SDP}
\langle \mathbf{Y} , \mathbf{A}_k \rangle = b_k,~~~ k = 1, \cdots , | \Omega |,
\end{align}
and thus~\eqref{eq:eqSDP_003} can be reformulated as
\begin{equation} \label{eq:eqSDP_004}
\begin{split}
\min\limits_{\mathbf{Y},t} &~~~~~ 2t \\
\text{subject to} &~~~~~ \text{tr}(\mathbf{Y}) \leq 2t  \\
&~~~~~ \mathbf{Y} \succeq 0  \\
&~~~~~ \langle \mathbf{Y}, \mathbf{A}_{k} \rangle = b_{k},~~ k = 1, \cdots, |\Omega|.
\end{split}
\end{equation}
For example, we consider the case where the desired matrix $\mathbf{M}$ is given by $\mathbf{M}= \begin{bmatrix}
1 & 2 \\ 
2 & 4
\end{bmatrix}$ and the index set of observed entries is $\Omega = \{(2,1),(2,2)\}$. In this case, 
\begin{align}
\mathbf{A}_{1}
= \mathbf{e}_{2} \mathbf{e}_{3}^{T},~
\mathbf{A}_{2}
= \mathbf{e}_{2} \mathbf{e}_{4}^{T},~
b_{1} = 2,~
\text{and}~ b_{2} = 4.
\end{align}
One can express~\eqref{eq:eqSDP_004} in a concise form as \eqref{eq:eqSDP_final}, which is the desired result.

\endproof

\section{Performance guarantee of NNM} \label{app:appB}

\proof[Sketch of proof]
Exact recovery of the desired low-rank matrix $\mathbf{M}$ can be guaranteed under the uniqueness condition of the NNM problem~\cite{candes_recht,candes_tao,recht}. To be specific, let $\mathbf{M} = \mathbf{U}\boldsymbol\Sigma\mathbf{V}^T$ be the SVD of $\mathbf{M}$ where $\mathbf{U}\in\mathbb{R}^{n_1\times r}$, $\boldsymbol\Sigma\in\mathbb{R}^{r\times r}$, and $\mathbf{V}\in\mathbb{R}^{n_2\times r}$. Also, let $\mathbb{R}^{n_1\times n_2} = T\bigoplus T^\perp$ be the orthogonal decomposition in which $T^\perp$ is defined as the subspace of matrices whose row and column spaces are orthogonal to the row and column spaces of $\mathbf{M}$, respectively. Here, $T$ is the orthogonal complement of $T^\perp$. It has been shown that $\mathbf{M}$ is the unique solution of the NNM problem if the following conditions hold true~\cite[Lemma 3.1]{candes_recht}:
\begin{itemize}
\item[1)] there exists a matrix $\mathbf{Y} = \mathbf{UV}^T + \mathbf{W}$ such that $P_\Omega(\mathbf{Y}) = \mathbf{Y}$, $\mathbf{W}\in T^{\perp}$, and $\|\mathbf{W}\| < 1$,
\item[2)] the restriction of the sampling operation $P_\Omega$ to $T$ is an injective (one-to-one) mapping.
\end{itemize}
The establishment of $\mathbf{Y}$ obeying 1) and 2) are in turn conditioned on the observation model of $\mathbf{M}$ and its intrinsic coherence property.

Under a uniform sampling model of $\mathbf{M}$, suppose the coherence property of $\mathbf{M}$ satisfies 
\begin{subequations}
\begin{align}
& \max(\mu(\mathbf{U}),\mu(\mathbf{V}))  \leq \mu_0,\\
& \max\limits_{ij}|e_{ij}|  \leq \mu_1\sqrt{\frac{r}{n_1n_2}},
\end{align}
\end{subequations}
where $\mu_0$ and $\mu_1$ are some constants, $e_{ij}$ is the entry of $\mathbf{E} = \mathbf{UV}^T$, and $\mu(\mathbf{U})$ and $\mu(\mathbf{V})$ are the coherences of the column and row spaces of $\mathbf{M}$, respectively.
\begin{theorem}[{\cite[Theorem 1.3]{candes_recht}}]
There exists constants $\alpha$ and $\beta$ such that if the number of observed entries $m = |\Omega|$ satisfies
\begin{align}
m \geq \alpha\max(\mu_1^2, \mu_0^{\frac{1}{2}}\mu_1,\mu_0n^{\frac{1}{4}})\gamma nr\log n
\label{eq:eq209}
\end{align}
where $\gamma > 2$ is some constant and $n_1 = n_2 = n$, then $\mathbf{M}$ is the unique solution of the NNM problem with probability at least $1 - \beta n^{-\gamma}$. Further, if $r\leq \mu_0^{-1}n^{1/5}$, \eqref{eq:eq209} can be improved to $m\geq C\mu_0\gamma n^{1.2}r\log n$ with the same success probability. 
\end{theorem} 
One direct interpretation of this theorem is that the desired low-rank matrix can be reconstructed exactly using NNM with overwhelming probability even when $m$ is much less than $n_1n_2$.


%
\endproof



\begin{thebibliography}{10}

\bibitem{netflix}
Netflix Prize. 
\newblock http://www.netflixprize.com

\bibitem{pal2010}
A.~Pal,
\newblock ``Localization algorithms in wireless sensor networks: Current approaches and future challenges,''
\newblock Netw. Protocols Algorithms, vol. 2, no. 1, pp. 45--74, 2010.

\bibitem{luongITA}
L.~Nguyen,~S.~Kim,~and~B.~Shim, 
\newblock ``Localization in Internet of things network: Matrix completion approach,''
\newblock in Proc. Inform. Theory Appl. Workshop, San Diego, CA, USA, 2016, pp. 1--4. 

\bibitem{luong2019}
L.~T.~Nguyen,~J.~Kim,~S.~Kim,~and~B. Shim,
\newblock ``Localization of IoT Networks Via Low-Rank Matrix Completion,''
\newblock to appear in IEEE Trans. Commun., 2019.

\bibitem{candes2015}
E. J. Candes,~Y. C. Eldar, and T. Strohmer,
\newblock ``Phase retrieval via matrix completion,''
\newblock SIAM Rev., vol. 52, no. 2, pp. 225--251, May 2015.

\bibitem{delamo2015}
M. Delamom,~S. Felici-Castell,~J. J. Perez-Solano, and A. Foster,
\newblock ``Designing an open source maintenance-free environmental monitoring application for wireless sensor networks,''
\newblock J. Syst. Softw., vol. 103, pp. 238--247, May 2015.

\bibitem{hackmann2014}
G.~Hackmann,~W.~Guo,~G.~Yan,~Z.~Sun,~C.~Lu, and S.~Dyke,
\newblock ``Cyber-physical codesign of distributed structural health monitoring with wireless sensor networks,''
\newblock IEEE Trans. Parallel Distrib. Syst., vol. 25, no. 1, pp. 63--72, Jan. 2014.

\bibitem{torgerson1952}
W. S. Torgerson, 
\newblock ``Multidimensional scaling: I. Theory and method,''
\newblock Psychometrika, vol. 17, no. 4, pp. 401--419, Dec. 1952.

\bibitem{hyoungju2017}
H.~Ji,~Y.~Kim,~J.~Lee,~E.~Onggosanusi,~Y.~Nam,~J.~Zhang,~B.~Lee, and B.~Shim,
\newblock ``Overview of full-dimension MIMO in LTE-advanced pro,''
\newblock IEEE Commun. Mag., vol. 55, no. 2, pp. 176--184, Feb. 2017.

\bibitem{marzetta2006}
T. L. Marzetta and B. M. Hochwald,
\newblock ``Fast transfer of channel state information in wireless systems,''
\newblock IEEE Trans. Signal Process., vol. 54, no. 4, pp. 1268--1278, Apr. 2006.

\bibitem{rusek2013}
F. Rusek,~D. Persson,~B. K. Lau,~E. G. Larsson,~T. L. Marzetta,~O. Edfors, and F. Tufvesson,
\newblock ``Scaling up MIMO: Opportunities and challenges with very large arrays,''
\newblock IEEE Signal Process. Mag., vol. 30, no. 1, pp. 40--60, Jan. 2013.

\bibitem{shen2015}
W. Shen,~L. Dai,~B. Shim,~S. Mumtaz, and Z. Wang,
\newblock ``Joint CSIT acquisition based on low-rank matrix completion for FDD massive MIMO systems,''
\newblock IEEE Commun. Lett., vol. 19, no. 12, pp. 2178--2181, Dec. 2015.
 
\bibitem{rappaport2013}
T.~S.~Rappaport~et~al.,
\newblock ``Millimeter wave mobile communications for 5G cellular: It will work!,''
\newblock IEEE Access, vol. 1, no. 1, pp. 335--349, May 2013.

\bibitem{li2018}
X.~Li,~J.~Fang,~H.~Li,~H.~Li,~and~P.~Wang,
\newblock ``Millimeter wave channel estimation via exploiting joint sparse and low-rank structures,''
\newblock IEEE Trans. Wireless Commun., vol. 17, no. 2, pp. 1123--1133, Feb. 2018.



\bibitem{shi2016}
Y.~Shi,~J.~Zhang,~and~K.~B.~Letaief,
\newblock ``Low-rank matrix completion for topological interference management by Riemannian pursuit,''
\newblock IEEE Trans. Wireless Commun., vol. 15, no. 7, pp. 4703-4717, Jul. 2016.

\bibitem{shi2017}
Y.~Shi,~B.~Mishra,~and~W.~Chen,
\newblock ``Topological interference management with user admission control via Riemannian optimization,''
\newblock IEEE Trans. Wireless Commun., vol. 16, no. 11, pp. 7362-7375, Nov. 2017.

\bibitem{shi2018}
Y.~Shi,~J.~Zhang,~W.~Chen,~and~K.~B.~Letaief,
\newblock ``Generalized sparse and low-rank optimization for ultra-dense networks,''
\newblock IEEE Commun. Mag., vol. 56, no. 6, pp. 42-48, Jun., 2018.

\bibitem{srid2015}
G.~Sridharan~and~W.~Yu,
\newblock ``Linear Beamforming Design for Interference Alignment via Rank Minimization,''
\newblock IEEE Trans. Signal Process., vol. 63, no. 22, pp. 5910-5923, Nov. 2015.

\bibitem{peng2016}
M.~Peng,~S.~Yan,~K.~Zhang,~and~C.~Wang,
\newblock ``Fog-computing-based radio access networks: issues and challenges,''
\newblock IEEE Network, vol. 30, pp. 46-53, July 2016.

\bibitem{kyang2016}
K.~Yang,~Y.~Shi,~and~Z.~Ding,
\newblock ``Low-rank matrix completion for mobile edge caching in Fog-RAN via Riemannian optimization,''
\newblock in Proc. IEEE Global Communications Conf. (GLOBECOM), Washington, DC, Dec. 2016.



\bibitem{fazel}
M.~Fazel,
\newblock ``Matrix rank minimization with applications,''
\newblock Ph.D. dissertation, Elec. Eng. Dept., Standford Univ., Stanford, CA, 2002.

\bibitem{candes_recht}
E.~J.~Candes and B.~Recht,
\newblock ``Exact matrix completion via convex optimization,''
\newblock Found. Comput. Math., vol. 9, no. 6, pp. 717--772, Dec. 2009.

\bibitem{candes_tao}
E.~J.~Candes and T.~Tao,
\newblock ``The power of convex relaxation: Near-optimal matrix completion,''
\newblock IEEE Trans. Inform. Theory, vol. 56, no. 5, pp. 2053--2080, May 2010.

\bibitem{toh}
K. C. Toh,~M. J. Todd, and R. H. Tutuncu,
\newblock ``SDPT3 --- a MATLAB software package for semidefinite programming,''
\newblock Optim. Methods Softw., vol. 11, pp. 545--581, 1999.

\bibitem{sturm1999using}
J.~F.~Sturm,
\newblock ``Using SeDuMi 1.02, a MATLAB toolbox for optimization over symmetric cones,''
\newblock Optim. Methods Softw., vol. 11, pp. 625--653, 1999.

\bibitem{van1996}
L.~Vandenberghe~and~S.~Boyd,
\newblock ``Semidefinite programming,''
\newblock SIAM Rev., vol. 38, no. 1, pp. 49--95, 1996.

\bibitem{zhang1998}
Y.~Zhang,
\newblock ``On extending some primal-dual interior-point algorithms from linear programming to semidefinite programming,''
\newblock SIAM J. Optim., vol. 8, no. 2, pp. 365--386, 1998.

\bibitem{nesterov1998}
Y.~E.~Nesterov~and~M.~Todd,
\newblock ``Primal-dual interior-point methods for self-scaled cones,''
\newblock SIAM J. Optim., vol. 8, no. 2, pp. 324--364, 1998.
 
\bibitem{potra2000}
F.~A.~Potra~and~S.~J.~Wright,
\newblock ``Interior-point methods,''
\newblock J. Comput. Appl. Math., vol. 124, no. 1-2, pp.281--302, 2000.


\bibitem{van2005}
L.~Vandenberghe,~V.~R.~Balakrishnan,~R.~Wallin,~A.~Hansson,~and~T.~Roh,
\newblock ``Interior-point algorithms for semidefinite programming problems derived from the KYP lemma,''
\newblock In Positive polynomials in control (pp. 195-238). Berlin, Heidelberg: Springer, 2005.


\bibitem{potra1998}
F.~A.~Potra~and~R.~Sheng,
\newblock ``A superlinearly convergent primal-dual infeasible-interior-point algorithm for semidefinite programming,''
\newblock SIAM J. Optim., vol. 8, no. 4, pp.1007--1028, 1998.

\bibitem{recht2010}
B.~Recht,~M.~Fazel,~and~P.~A.~Parillo,
\newblock ``Guaranteed minimum-rank solutions of linear matrix equations via nuclear norm minimization,"
\newblock SIAM Rev., vol. 52, no. 3, pp. 471--501, 2010.

\bibitem{svt}
J. F. Cai,~E. J. Candes, and Z. Shen,
\newblock ``A singular value thresholding algorithm for matrix completion,''
\newblock SIAM J. Optim., vol. 20, no. 4, pp. 1956--1982, Mar. 2010.

\bibitem{blumensath}
T. Blumensath and M. E. Davies,
\newblock ``Iterative hard thresholding for compressed sensing,''
\newblock Appl. Comput. Harmon. Anal., vol. 27, no. 3, pp. 265--274, Nov. 2009.

\bibitem{tanner}
J. Tanner and K. Wei,
\newblock ``Normalized iterative hard thresholding for matrix completion,''
\newblock SIAM J. Sci. Comput., vol. 35, no. 5, pp. S104--S125, Oct. 2013.

\bibitem{fornasier2011}
M.~Fornasier,~H.~Rauhut,~and~R.~Ward,
\newblock ``Low-rank matrix recovery via iteratively reweighted least squares minimization,"
\newblock SIAM J. Optim., vol. 21, no. 4, pp. 1614--1640, Dec. 2011.

\bibitem{mohan2012}
K.~Mohan,~and~M.~Fazel,
\newblock ``Iterative reweighted algorithms for matrix rank minimization,"
\newblock J. Mach. Learning Research, no. 13, pp. 3441--3473, Nov. 2012.


\bibitem{needell2009}
D. Needell and J. A. Tropp,
\newblock ``CoSaMP: Iterative signal recovery from incomplete and inaccurate samples,''
\newblock Appl. Comput. Harmon. Anal., vol. 26, no. 3, pp. 301--321, May 2009. 

\bibitem{choi2017}
J.~W.~Choi,~B.~Shim,~Y.~Ding,~B.~Rao,~and~D.~I.~Kim,
\newblock ``Compressed sensing for wireless communications: Useful tips and tricks,''
\newblock IEEE Commun. Surveys Tuts., vol. 19, no. 3, pp. 1527--1550, Feb. 2017.

\bibitem{kwon2014}
S.~Kwon,~J.~Wang,~and~B. Shim,
\newblock ``Multipath matching pursuit,''
\newblock IEEE Trans. Inform. Theory, vol. 60, no. 5, pp. 2986--3001, Mar. 2014.

\bibitem{wang2012}
J.~Wang,~S.~Kwon,~and~B.~Shim, 
\newblock ``Generalized orthogonal matching pursuit,''
\newblock IEEE Trans. Signal Process., vol. 60, no. 12, pp. 6202--6216, Sep. 2012.


\bibitem{tropp2007}
J. A. Tropp and A. C. Gilbert, 
\newblock ``Signal recovery from random measurements via orthogonal matching pursuit,''
\newblock IEEE Trans. Inform. Theory, vol. 53, no. 12, pp. 4655--4666, Dec. 2007.

\bibitem{adm}
K. Lee and Y. Bresler,
\newblock ``ADMiRA: Atomic decomposition for minimum rank approximation,''
\newblock IEEE Trans. Inform. Theory, vol. 56, no. 9, pp. 4402--4416, Sept. 2010.

\bibitem{rom}
Z. Wang,~M-J. Lai,~Z. Lu,~W. Fan,~H. Davulcu, and J. Ye,
\newblock ``Rank-one matrix pursuit for matrix completion,''
\newblock in Proc. Int. Conf. Mach. Learn., Beijing, China, 2014, pp. 91--99.

\bibitem{pf}
J. P. Haldar and D. Hernando, 
\newblock ``Rank-constrained solutions to linear matrix equations using power factorization,''
\newblock IEEE Signal Process. Lett., vol. 16, no. 7, pp. 584--587, Jul. 2009.

\bibitem{asd}
J. Tanner and ~K. Wei, 
\newblock ``Low rank matrix completion by alternating steepest descent methods,''
\newblock Appl. Comput. Harmon. Anal., vol. 40, no. 2, pp. 417--429, Mar. 2016.

\bibitem{lmafit}
Z. Wen,~W. Yin, and Y. Zhang, 
\newblock ``Solving a low-rank factorization model for matrix completion by a nonlinear successive over-relaxation algorithm,'' 
\newblock Math. Prog. Comput., vol. 4, no. 4, pp. 333--361, Dec. 2012.




\bibitem{mishra2014}
B. Mishra, G. Meyer, S. Bonnabel, and R. Sepulchre,
\newblock ``Fixed-rank matrix factorizations and Riemannian low-rank optimization,'' 
\newblock Comput. Stat., vol. 3, no. 4, pp. 591--621, Jun. 2014.

\bibitem{dai2010}
W. Dai and O. Milenkovic, 
\newblock ``SET: An algorithm for consistent matrix completion,''
\newblock in Proc. Int. Conf. Acoust., Speech, Signal Process., Dallas, Texas, USA, 2010, pp. 3646--3649.

\bibitem{bart2013}
B. Vandereycken, 
\newblock ``Low-rank matrix completion by Riemannian optimization,'' 
\newblock SIAM J. Optim., vol. 23, no. 2, pp. 1214--1236, Jun. 2013.

\bibitem{sgg}
T. Ngo and Y. Saad, 
\newblock ``Scaled gradients on Grassmann manifolds for matrix completion,''
\newblock in Proc. Adv. Neural Inform. Process. Syst. Conf., Lake Tahoe, Nevada, USA, 2012, pp. 1412--1420.

\bibitem{dattorro2005}
J.~Dattorro,
\newblock Convex optimization and Euclidean distance geometry. USA: Meboo, 2005. 

\bibitem{helmke1994}
U.~Helmke and J.~B.~Moore,
\newblock Optimization and Dynamical Systems. New York, NY, USA: Springer, 1994.

\bibitem{bart2009}
B. Vandereycken,~P.-A. Absil, and S. Vandewalle, 
\newblock ``Embedded geometry of the set of symmetric positive semidefinite matrices of fixed rank,'' 
\newblock in Proc. IEEE Workshop Stat. Signal Process., Cardiff, UK, 2009, pp. 389--392.

\bibitem{absil2009optimization}
P.~A.~Absil,~R.~Mahony, and R.~Sepulchre,
\newblock Optimization algorithms on matrix manifolds. Princeton, NJ, USA: Princeton Univ., 2009.


\bibitem{lee2003smooth}
J.~M.~Lee,
\newblock Smooth manifolds. New York, NY, USA: Springer, 2003.



\bibitem{truncatedNNM}
Y. Hu,~D. Zhan,~J. Ye,~X. Li, and X. He, 
\newblock ``Fast and accurate matrix completion via truncated nuclear norm regularization,''
\newblock IEEE Trans. Pattern Anal. Mach. Intell., vol. 35, no. 9, pp. 2117--2130, Sept. 2013.


\bibitem{DCA}
J.~Y.~Gotoh, A. Takeda, and K. Tono, 
\newblock ``DC formulations and algorithms for sparse optimization problems,''
\newblock Math. Programming, pp. 1-36, May 2018.

\bibitem{ge2016}
R.~Ge,~J.~D.~Lee,~and~T.~Ma, 
\newblock ``Matrix completion has no spurious local minimum,'' 
\newblock in Advances Neural Inform. Process. Syst., pp. 2973-2981, 2016.


\bibitem{ge2017}
R.~Ge,~C.~Jin,~and~Y.~Zheng, 
\newblock ``No spurious local minima in nonconvex low rank problems: A unified geometric analysis,''
\newblock In Proc. 34th Int. Conf. on Machine Learning, JMLR. org., Aug. 2017, vol. 70, pp. 1233--1242.


\bibitem{du2017}
S.~S.~Du,~C.~Jin,~J.~D.~Lee,~M.~I.~Jordan,~A.~Singh,~and~B.~Poczos,
\newblock ``Gradient descent can take exponential time to escape saddle points,'' 
\newblock in Advances Neural Inform. Process. Syst., pp. 1067-1077, 2017.



\bibitem{scarselli2009}
F. Scarselli,~M. Gori,~A. C. Tsoi,~M. Hagenbuchner, and G. Monfardini,
\newblock ``The graph neural network model,''
\newblock IEEE Trans. Neural Netw., vol. 20, no. 1, pp. 61--80, Jan. 2009.

\bibitem{hammond2011}
D.~K.~Hammond,~P.~Vandergheynst, and R.~Gribonval,
\newblock ``Wavelets on graphs via spectral graph theory,''
\newblock Appl. Comput. Harmon. Anal., vol. 30, no. 2, pp. 129--150, Mar. 2011.

\bibitem{sedhain2015autorec}
S.~Sedhain,~A.~K.~Menon,~S.~Sanner, and L.~Xie,
\newblock ``Autorec: Autoencoders meet collaborative filtering,''
\newblock in Proc. Int. Conf. World Wide Web, Florence, Italy, 2015, pp. 111--112.

\bibitem{zheng2016neural}
Y.~Zheng,~B.~Tang,~W.~Ding, and H.~Zhou,
\newblock ``A neural autoregressive approach to collaborative filtering,''
\newblock in Proc. Int. Conf. Mach. Learn., New York, NY, USA, 2016, pp. 764--773.

\bibitem{he2017neural}
X.~He,~L.~Liao,~H.~Zhang,~L.~Nie,~X.~Hu, and T.~S.~Chua,
\newblock ``Neural collaborative filtering,''
\newblock in Proc. Int. Conf. World Wide Web, Perth, Australia, 2017, pp. 173--182.

\bibitem{monti2017}
F. Monti,~M. Bronstein, and X. Bresson,
\newblock ``Geometric matrix completion with recurrent multi-graph neural networks,''
\newblock in Proc. Adv. Neural Inform. Process. Syst., Long Beach, CA, USA, 2017, pp. 3700--3710.

\bibitem{bruna2014spectral}
J.~Bruna,~W.~Zaremba,~A.~Szlam, and Y.~Lecun,
\newblock ``Spectral networks and locally connected networks on graphs,''
\newblock in Proc. Int. Conf. Learn. Representations, Banff, Canada, 2014, pp. 1--14.

\bibitem{henaff2015deep}
M.~Henaff,~J.~Bruna, and Y.~Lecun,
\newblock ``Deep convolutional networks on graph-structured data,''
\newblock arXiv:1506.05163, 2015.

\bibitem{defferrard2016}
M. Defferrard,~X. Bresson, and P. Vandergheynst,
\newblock ``Convolutional neural networks on graphs with fast localized spectral filtering,''
\newblock in Proc. Adv. Neural Inform. Process. Syst., Barcelona, Spain, 2016, pp. 3844--3852.

\bibitem{kipf2016}
T.~N.~Kipf~and~M.~Welling,
\newblock ``Semi-supervised classification with graph convolutional networks,''
\newblock arXiv preprint arXiv:1609.02907, 2016.

\bibitem{montigeo2017}
F.~Monti,~D.~Boscaini,~J.~Masci,~E.~Rodola,~J.~Svoboda,~and~M.~M.~Bronstein, 
\newblock ``Geometric deep learning on graphs and manifolds using mixture model cnns,''
\newblock in Proc. IEEE Conf. Comput. Vision Pattern Recognition, pp. 5115--5124, 2017.




\bibitem{davenport2016overview}
M.~A.~Davenport and J.~Romberg,
\newblock ``An overview of low-rank matrix recovery from incomplete observations,''
\newblock IEEE J. Sel. Topics Signal Process., vol. 10, no. 4, pp. 608--622, Jun. 2016.

\bibitem{chen2018harnessing}
Y.~Chen and Y.~Chi,
\newblock ``Harnessing structures in big data via guaranteed low-rank matrix estimation,''
\newblock IEEE Signal Process. Mag., vol. 35, no. 4, pp. 14--31, Jul. 2018.

\bibitem{recht}
B. Recht,
\newblock ``A simple approach to matrix completion,''
\newblock J. Mach. Learn. Res., vol. 12, pp. 3413--3430, Dec. 2011.

\bibitem{rankmin_complexity}
C. R. Berger,
\newblock Double Exponential.
\newblock IEEE Trans. Signal Process. \textbf{56}(5), 1708--1721 (2010)

\bibitem{boyd}
S.~Boyd and Van,
\newblock Convex Optimization. Cambridge, England: Cambridge Univ., 2004. 


 

\bibitem{combettes}
P. Combettes and J. C. Pesquet,
\newblock Proximal splitting methods in signal processing. New York, NY, USA: Springer, 2011.

\bibitem{jain}
P. Jain,~R. Meka, and I. Dhillon,
\newblock ``Guaranteed rank minimization via singular value projection,'' 
\newblock in Proc. Neural Inform. Process. Syst. Conf., Vancouver, Canada, 2010, pp. 937--945.

\bibitem{tao2011}
M. Tao and X. Yuan,
\newblock ``Recovering low-rank and sparse components of matrices from incomplete and noisy observations,''
\newblock SIAM J. Optim., vol. 21, no. 1, pp. 57--81, Jan. 2011.

\bibitem{lin2011}
Z. Lin,~R. Liu, and Z. Su,
\newblock ``Linearized alternating direction method with adaptive penalty for low-rank representation,''
\newblock in Proc. Adv. Neural Inform. Process. Syst., Montreal, Canada, 2011, pp. 612--620. 

\bibitem{he2000}
B.~S.~He,~H.~Yang, and S.~L.~Wang, 
\newblock ``Alternating direction method with self-adaptive penalty parameters for monotone variational inequalities,'' 
\newblock J. Optim. Theory Appl., vol. 106, no. 2, pp. 337--356, Aug. 2000.

\bibitem{beck2009}
A. Beck and M. Teboulle, 
\newblock ``A fast iterative shrinkage-thresholding algorithm for linear inverse problems,''
\newblock SIAM J. Imaging Sci., vol. 2, no. 1, pp. 183--202, Mar. 2009.



\bibitem{apm}
R. Escalante and M. Raydan,
\newblock Alternating projection methods. Philadelphia, PA, USA: SIAM, 2011.




\bibitem{chandrasekaran2012}
V. Chandrasekaran,~B. Recht,~P. A. Parrilo, and A. S. Willsky,
\newblock ``The convex geometry of linear inverse problems,''
\newblock Found. Comput. Math., vol. 12, no. 6, pp. 805--849, Dec. 2012.

\bibitem{bhaskar2013}
B. N. Bhaskar,~G. Tang, and B. Recht,
\newblock ``Atomic norm denoising with applications to line spectral estimation,''
\newblock IEEE Trans. Signal Process., vol. 61, no. 23, pp. 5987--5999, Dec. 2013.

\bibitem{li2016off}
Y. Li and Y. Chi,
\newblock ``Off-the-grid line spectrum denoising and estimation with multiple measurement vectors,''
\newblock IEEE Trans. Signal Process., vol. 64, no. 5, pp. 1257--1269, Mar. 2016.


\bibitem{chen2001}
S. S. Chen,~D. L. Donoho, and M. A. Saunders,
\newblock ``Atomic decomposition by basis pursuit,''
\newblock SIAM Rev., vol. 43, no. 1, pp. 129--159, Feb. 2001.

\bibitem{rao2015}
N. Rao,~P. Shah, and S. Wright,
\newblock ``Forward–backward greedy algorithms for atomic norm regularization,''
\newblock IEEE Trans. Signal Process., vol. 63, no. 21, pp. 5798--5811, Nov. 2015.

\bibitem{donoho2009}
D. L. Donoho, A. Maleki, and A. Montanari, 
\newblock ``Message passing algorithms for compressed sensing,'' 
\newblock Proc. Nat. Acad. Sci., vol. 106, no. 45, pp. 18914--18919, Nov. 2009.

\bibitem{goodfellow2016}
I.~Goodfellow,
\newblock ``NIPS 2016 tutorial: Generative adversarial networks,''
\newblock arXiv preprint arXiv:1701.00160, 2016.





\end{thebibliography}
\end{document}